%% file: main.tex
\documentclass[copyright,creativecommons]{eptcs}
\usepackage[utf8]{inputenc}

\usepackage[pdftex,dvipsnames]{xcolor}  
\usepackage{amsmath,amsthm,amssymb,amsfonts,graphicx}
\usepackage{stmaryrd}
\usepackage{tikz}
\usepackage{proof}
\usepackage{enumerate}
\usepackage{mathtools}
\usepackage{xypic}
\usepackage{lscape}
\usepackage{multicol} 
\usepackage{leftidx}
\usepackage{bbm}
\usepackage{rotating}
\usepackage{extarrows}
\usepackage{dsfont}
\providecommand{\event}{ACT 2022}

\input{tikz.tex}

\newcommand{\X}{\mathbb{X}}
\newcommand{\A}{\mathbb{A}}

\newcommand{\C}{\mathbb{D}}
\newcommand{\D}{\mathbb{C}}

\newcommand{\Y}{\mathbb{Y}}

\newcommand{\R}{\mathbb{R}}

\newcommand{\f}{{\sf f}}

\newcommand{\ox}{\otimes}

\newcommand{\op}{\mathsf{op}}

\newcommand{\priyaa}[1]{\textcolor{purple}{#1}}

\newcommand{\chaos}{\mathsf{chaos}}
\newcommand{\order}{\mathsf{order}}

\renewcommand{\phi}{\varphi}

\newcommand{\Rand}{{\sf Rand}}
\newcommand{\Uniform}{{\sf Uniform}}
\newcommand{\Bip}{{\sf Bip}}
\newcommand{\PureBip}{{\sf PureBip}}
\newcommand{\LOCC}{{\sf LOCC}}
\newcommand{\LOCCp}{{\sf LOCC_p}}
\newcommand{\Divergence}{{\sf Distinguish}}
\newcommand{\cDivergence}{{\sf cDistinguish}}
\newcommand{\Processing}{{\sf Processing}}
\newcommand{\cProcessing}{{\sf cProcessing}}
\newcommand{\qRand}{{\sf qRand}}
\newcommand{\qUniform}{{\sf qUniform}}
\newcommand{\pCat}[2]{\left({\sf chaos}_{#1}, #2_{#1}\right)}
\newcommand{\Schmidt}{{\sf Schmidt}}

\newcommand{\chan}[1]{\mathcal{#1}}
\newcommand{\maxmix}[1]{\frac{1}{{\sf dim}(#1) } 1_{#1}}

\newcommand{\Shannon}{\mathsf{Shannon}}

\newcounter{dummy} 
\numberwithin{dummy}{section}

\newtheorem{lemma}[dummy]{Lemma}  
\newtheorem{theorem}[dummy]{Theorem}
\theoremstyle{definition}
\newtheorem{definition}[dummy]{Definition}

\newtheorem{corollary}[dummy]{Corollary} 
\theoremstyle{definition}
\newtheorem{example}[dummy]{Example}
\numberwithin{equation}{section}

\newlength{\llcfoo}



\makeatletter

\newdimen\w@dth

\def\setw@dth#1#2{\setbox\z@\hbox{\scriptsize $#1$}\w@dth=\wd\z@
\setbox\@ne\hbox{\scriptsize $#2$}\ifnum\w@dth<\wd\@ne \w@dth=\wd\@ne \fi
\advance\w@dth by 1.2em}

\def\t@^#1_#2{\allowbreak\def\n@one{#1}\def\n@two{#2}\mathrel
{\setw@dth{#1}{#2}
\mathop{\hbox to \w@dth{\rightarrowfill}}\limits
\ifx\n@one\empty\else ^{\box\z@}\fi
\ifx\n@two\empty\else _{\box\@ne}\fi}}
\def\t@@^#1{\@ifnextchar_ {\t@^{#1}}{\t@^{#1}_{}}}

\def\t@left^#1_#2{\def\n@one{#1}\def\n@two{#2}\mathrel{\setw@dth{#1}{#2}
\mathop{\hbox to \w@dth{\leftarrowfill}}\limits
\ifx\n@one\empty\else ^{\box\z@}\fi
\ifx\n@two\empty\else _{\box\@ne}\fi}}
\def\t@@left^#1{\@ifnextchar_ {\t@left^{#1}}{\t@left^{#1}_{}}}

\def\two@^#1_#2{\def\n@one{#1}\def\n@two{#2}\mathrel{\setw@dth{#1}{#2}
\mathop{\vcenter{\hbox to \w@dth{\rightarrowfill}\kern-1.7ex
                 \hbox to \w@dth{\rightarrowfill}}%
       }\limits
\ifx\n@one\empty\else ^{\box\z@}\fi
\ifx\n@two\empty\else _{\box\@ne}\fi}}
\def\tw@@^#1{\@ifnextchar_ {\two@^{#1}}{\two@^{#1}_{}}}

\def\tofr@^#1_#2{\def\n@one{#1}\def\n@two{#2}\mathrel{\setw@dth{#1}{#2}
\mathop{\vcenter{\hbox to \w@dth{\rightarrowfill}\kern-1.7ex
                 \hbox to \w@dth{\leftarrowfill}}%
       }\limits
\ifx\n@one\empty\else ^{\box\z@}\fi
\ifx\n@two\empty\else _{\box\@ne}\fi}}
\def\t@fr@^#1{\@ifnextchar_ {\tofr@^{#1}}{\tofr@^{#1}_{}}}

\newdimen\W@dth
\def\setW@dth#1#2{\setbox\z@\hbox{$#1$}\W@dth=\wd\z@
\setbox\@ne\hbox{$#2$}\ifnum\W@dth<\wd\@ne \W@dth=\wd\@ne \fi
\advance\W@dth by 1.2em}

\def\T@^#1_#2{\allowbreak\def\N@one{#1}\def\N@two{#2}\mathrel
{\setW@dth{#1}{#2}
\mathop{\hbox to \W@dth{\rightarrowfill}}\limits
\ifx\N@one\empty\else ^{\box\z@}\fi
\ifx\N@two\empty\else _{\box\@ne}\fi}}
\def\T@@^#1{\@ifnextchar_ {\T@^{#1}}{\T@^{#1}_{}}}

\def\T@left^#1_#2{\def\N@one{#1}\def\N@two{#2}\mathrel{\setW@dth{#1}{#2}
\mathop{\hbox to \W@dth{\leftarrowfill}}\limits
\ifx\N@one\empty\else ^{\box\z@}\fi
\ifx\N@two\empty\else _{\box\@ne}\fi}}
\def\T@@left^#1{\@ifnextchar_ {\T@left^{#1}}{\T@left^{#1}_{}}}

\def\Tofr@^#1_#2{\def\N@one{#1}\def\N@two{#2}\mathrel{\setW@dth{#1}{#2}
\mathop{\vcenter{\hbox to \W@dth{\rightarrowfill}\kern-1.7ex
                 \hbox to \W@dth{\leftarrowfill}}%
       }\limits
\ifx\N@one\empty\else ^{\box\z@}\fi
\ifx\N@two\empty\else _{\box\@ne}\fi}}
\def\T@fr@^#1{\@ifnextchar_ {\Tofr@^{#1}}{\Tofr@^{#1}_{}}}

\def\Two@^#1_#2{\def\N@one{#1}\def\N@two{#2}\mathrel{\setW@dth{#1}{#2}
\mathop{\vcenter{\hbox to \W@dth{\rightarrowfill}\kern-1.7ex
                 \hbox to \W@dth{\rightarrowfill}}%
       }\limits
\ifx\N@one\empty\else ^{\box\z@}\fi
\ifx\N@two\empty\else _{\box\@ne}\fi}}
\def\Tw@@^#1{\@ifnextchar_ {\Two@^{#1}}{\Two@^{#1}_{}}}

\def\to{\@ifnextchar^ {\t@@}{\t@@^{}}}
\def\from{\@ifnextchar^ {\t@@left}{\t@@left^{}}}
\def\tofro{\@ifnextchar^ {\t@fr@}{\t@fr@^{}}}
\def\To{\@ifnextchar^ {\T@@}{\T@@^{}}}
\def\From{\@ifnextchar^ {\T@@left}{\T@@left^{}}}
\def\Two{\@ifnextchar^ {\Tw@@}{\Tw@@^{}}}
\def\Tofro{\@ifnextchar^ {\T@fr@}{\T@fr@^{}}}

\makeatother

\begin{document}

\title{Extending Resource Monotones using Kan Extensions}
\author{Robin Cockett\institute{Department of Computer Science, University of Calgary, Alberta, Canada}\institute{Institute for Quantum Science and Technology, University of Calgary, Alberta, Canada}\and Isabelle Jianing Geng \qquad\qquad Carlo Maria Scandolo\institute{Department of Mathematics \& Statistics, University of Calgary, Alberta, Canada}\institute{Institute for Quantum Science and Technology, University of Calgary, Alberta, Canada}\and Priyaa Varshinee Srinivasan\institute{Department of Computer Science, University of Calgary, Alberta, Canada}\institute{National Institute of Standards and Technology, Maryland, USA}}
\def\titlerunning{Extending resource monotones using Kan Extensions}
\def\authorrunning{R.\ Cockett, I.\ J.\ Geng, C.\ M.\ Scandolo \& P.\ V.\ Srinivasan}

\maketitle 
\begin{abstract}
In this paper we generalize the framework proposed by Gour and Tomamichel regarding extensions of monotones for resource theories. A monotone for a resource theory assigns a real number to each resource in the theory signifying the utility or the value of the resource.  Gour and Tomamichel studied the problem of extending monotones using set-theoretical framework when a resource theory embeds fully and faithfully into the larger theory. One can generalize the problem of computing monotone extensions to scenarios when there exists a functorial transformation of one resource theory to another instead of just a full and faithful inclusion. In this article, we show that (point-wise) Kan extensions provide a precise categorical framework to describe and compute such extensions of monotones. To set up monontone extensions using Kan extensions, we introduce partitioned categories (pCat) as a framework for resource theories and pCat functors to formalize relationship between resource theories. We describe monotones as pCat functors into $([0,\infty], \leq)$, and describe extending monotones along any pCat functor using Kan extensions. We show how our framework works by applying it to extend entanglement monotones for bipartite pure states to bipartite mixed states, to extend classical divergences to the quantum setting, and to extend a non-uniformity monotone from classical probabilistic theory to quantum theory.

\end{abstract}

\section{Introduction}

Resource theories \cite{Quantum-resource-1,Quantum-resource-2,Gour-review} in physics model systems in which certain operations considered to be `free of cost' among of the set of all operations. For example, placing a glass of chilled
water at room temperature warms up the water to the ambient temperature. In this
context, operations that change the temperature of the water to be in equilibrium with the ambient temperature are considered to be free.   In order to produce a
``resourceful state"  --- for example, a glass of chilled water ---  one
requires non-free transformations, such as a fridge, which consumes electricity. Resource theories have been successfully used to
study, among other examples, thermodynamical systems \cite{delRio,Lostaglio-thermo}, entanglement \cite{Review-entanglement,Dynamical-entanglement}, and coherence \cite{Review-coherence}. 

A central question in the resource-theoretic modelling of systems is:  \textit{given two resources, is there a free transformation to convert one resource into the other?} The answer to this question imposes a preorder on resources which captures their value or usefulness. Intuitively, a resource is more valuable than another if, by possessing the former, we are given access to a larger set of resources including the latter through free transformations. This not a partial order, because their may be different resources that can be converted freely into each other. Such resources are considered equivalent. In this way, we can set up a partial order on the equivalence classes of resources. One way to define such an order is to quantify resources by introducing  monotones, which are a order-preserving maps from the set of all resources into $[0,\infty]$ \cite{Gour-review}. Monotones assign a value to resources that is compatible with the preorder, viz.\ with their usefulness. Monotones often have a physical meaning, such as in the resource theories of quantum thermodynamics \cite{Lostaglio-thermo}, where, for systems at a fixed temperature, free energy is a monotone, and for isolated systems, entropy is the natural monotone.

Given a monotone $M$ for a resource theory which embeds in a larger theory, a natural question to ask is whether the monotone M for the smaller theory can be used to quantify the resources in the larger resource theory. This question arises from the observation that that resources exclusive to the larger theory can possibly be converted to resources contained in the smaller theory, and vice versa. It turns out that one can always compute the optimal upper and lower bound for the value of every resource in the larger theory.  In other words, it is possible to extend the monotone M  to give optimal upper and lower bounds respectively on the value of resources in the larger theory.

In \cite{GoT20}, Gour and Tomamichel presented a set-theoretical framework for extending monotones from a subset of resources to the entire set of resources. A similar construction was also introduced by Gonda and Spekkens in \cite{Gonda}. Given a monotone $M$ over a subset of states, they  compute `minimal' and `maximal' extensions of the monotone to the entire set of states. 
In this article, we show that these extensions are special cases of  more general categorical concepts, called  (point-wise) left and right Kan extensions \cite{Kan58, Bor94, Mac13, Rie17}. 
Kan extensions deal with optimally extending a functor $F: \X \to \A$ along another functor $K: \Y \to \A$ to give two functors:  $\overline{F}_K: \Y \to A$ called the left Kan extension of $F$ along $K$, and $\underline{F}_K: \Y \to  \A$ called the right Kan extension of $F$ along $K$. 
The right Kan extension can be interpreted as the most conservative extension of $F$ along $K$ and the left Kan extension as the most liberal extension of $F$ along $K$.

We first introduce partitioned Categories (pCats) as a framework for resource theories. Partitioned categories are categories with a chosen subcategory of free transformations. The subcategory includes all the objects of the parent category, in other words, the inclusion of the subcategory into the parent category is bijective on objects. Relationships between resource theories are set up as pCat functors. In this article, since we consider monotones which are not necessarily additive, thus we do not demand symmetric monoidal structure on pCats. 

Given a resource theory, the necessary and sufficient conditions for transformations of resources can be encoded as a pCat functor from the resource theory into a preorder. We call such a pCat functor as a preorder collapse. A resource monotone is a preorder collapse into  $([0,\infty], \leq)$. In resource theories, contravariant rather than covariant monotones are encountered more frequently, the reason being if resource $A$ can be transformed to resource $B$ using only free transformation(s), then the value of $A$ is considered to be at least as high as the value of $B$. We refer to a contravariant resource monotone as an op-monotone. The distinction between monotones and op-monotones is important in the computation of monotone extensions. The categorical descriptions of pCats, pCat functors, preorder collapse, monotones are discussed with various running examples in Section~\ref{Sec: categorical framework}. Table~\ref{Table: functors}, we briefly summarizes the functors of resource theories introduced in this article.
\begin{table}
\centering
\begin{tabular}{ |c|c|c| } 
 \hline
 pCat functors & functors which preserve free transformations \\
 \hline
 Preorder collapse & a pCat functor whose codomain category is a preorder \\ 
 \hline
 Monotone & a pCat functor whose codomain category is $([0, \infty], \leq)$ \\ 
 \hline 
 Op-monotone & a pCat functor whose codomain category is $([0, \infty], \geq)$ \\
 \hline
\end{tabular}
'\caption{Functors for resource theories}
\label{Table: functors}
\end{table}

Having set up monotones as pCat functors, optimal extensions of monotones along any pCat functor are given by their left and right Kan extensions. Lemma~\ref{Lemma: extension properties} examines the properties of the monotone extensions thus computed, and prove that such extensions are optimal and monotonic. In Lemma~\ref{Lemma: extension properties ff}, we show that extending a monotone along a full and faithful functor recovers the case described in \cite{GoT20} by Gour and Tomamichel and by Gonda and Spekkens in \cite{Gonda} (therein called yield and cost constructions).  We apply the Kan extension framework for monotones to extend classical divergences to the quantum setting, to extend bipartite pure states entanglement monotone to mixed states, and to extend Shannon entropy as a measure of non-uniformity from classical probabilistic theory to quantum theory. Section~\ref{Sec: monotone extensions} is dedicated to setting up the Kan extension framework for monotones, and to studying the extension properties and its applications. 

{ \bf  Notation:} 
{\em In this paper, we use bold letters $\mathbb{X}$, $\mathbb{Y}$, $\C$ to denote categories. We use uppercase letters to denote both objects in the categories and functors between categories, whose meaning will be clear from the context. Lowercase letters $f, g, h,\pi$ are reserved for maps in the categories. Let $X, Y, Z$ be objects, and let $X\xrightarrow{f}Y$, $Y\xrightarrow{g}Z$ be two arrows, we denote the composition of the two arrows $X\xrightarrow{f}Y\xrightarrow{g}Z$ as $fg$, and similar notations apply for the composition of functors.}

\section{An introduction to Kan extensions}

 Kan extensions \cite{Kan58, Bor94, Mac13, Rie17} are a broadly applicable notion which is quite central to category theory. Indeed, Mac Lane in his book ‘Categories for working Mathematician’ \cite{Mac13} gave the chapter on Kan extensions the title “All concepts are Kan extensions”.  In this section, we provide the definition of Kan extensions and discuss limits and colimits as an example of Kan extensions.
 
\subsection{Left and right Kan extensions}

We first provide the definition Kan extensions of a functor along another functor, and explain the universal properties.

\begin{definition} 
Let $F: \X \to \C$ and $K: \X \to \Y$ be any two functors. 

\begin{enumerate}[(i)]
\item {\bf Right Kan (minimal) extension} of $F$ along $K$ is a functor $\underline{F}_K: \Y \to \C$ with a natural transformation $\underline{\psi}: K \underline{F}_K \Rightarrow F$  which is universal , see Fig.~\ref{Fig: universal extensions}-(a). The right Kan extension is written as $(\underline{F}_K, \underline{\psi})$.

\item {\bf Left Kan (maximal) extension} of $F$ along $K$ is a functor $\overline{F}_K: \Y \to \C$ with a natural transformation $\overline{\psi}: F \Rightarrow K \overline{F}_K $  which is couniversal, see Fig.~\ref{Fig: universal extensions}-(b). The left Kan extension is written as $(\overline{F}_K, \overline{\psi})$.
\end{enumerate}

\end{definition}

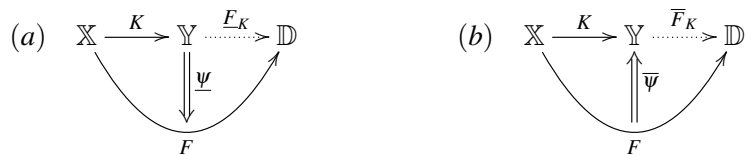
\begin{figure}[h]
\[ (a)~~~  \xymatrix{ \X \ar[r]^{K} \ar@/_3pc/[rr]_F & \Y \ar@{.>}[r]^{\underline{F}_K} \ar@{=>}[d]^{\underline{\psi}}  & \C \\ 
 & & &
  }
   ~~~~~~~~~~ 
 (b)~~~  \xymatrix{ \X \ar[r]^{K} \ar@/_3pc/[rr]_F & \Y \ar@{.>}[r]^{\overline{F}_K} \ar@{<=}[d]^{\overline{\psi}}  & \C \\ 
 & & &
  } \]
 \caption{(a) Right Kan Extension  \qquad \qquad (b) Left Kan Extension} 
 \label{Fig: Kan extensions}
\end{figure}

Fig.~\ref{Fig: Kan extensions} shows the Kan extensions of $F$ along $K$. We refer to the category $\C$ as the {\bf target}, the category $\X$ as the {\bf source categories}. Functor $F$ is extended from its source $\X$  along $K$. Right and Left Kan extensions of $F$ along $K$ are usually written as ${\rm Ran}_K(F)$ and ${\rm Lan}_K(F)$. However, we use the notation introduced in \cite{GoT20} for resource monotone extensions for uniformity.

Let us examine the universal properties of the Kan extensions.
The universal property of right Kan extension assures that for any other functor $H: \Y \to \C$ with a natural transformation $\gamma: KH \Rightarrow F$, there exists a $\gamma': H \to \underline{F}_{K}$ such that $\gamma$ factors through $\underline{\psi}$ via $\gamma'$, that is, $\gamma = (K \otimes \gamma') \underline{\psi}$ (See Fig.~\ref{Fig: universal extensions}-(a)). Informally, this means the right Kan extension of $F$ along $K$ is the most conservative extension and that any other extension $H$ can be transformed to $\underline{F}_K$. In this sense, $\underline{F}_K$ is the minimal extension of $F$ along $K$.

Similarly the couniversal property of left Kan extension assures that for any other functor $H: \Y \to \C$ with a natural transformation $\delta: F \Rightarrow KH$, there exists a $\delta': \overline{F}_{K} \to H$ such that $\delta$ factors through $\overline{\psi}$ via $\delta'$, that is, $\gamma = \overline{\psi} (K \otimes \gamma')$ (See Fig.~\ref{Fig: universal extensions}-(b)). Informally,  this means that $\overline{F}_K$ can be naturally transformed to any other such $H$. In this sense, $\overline{F}_K$ is the maximal extension of $F$ along $K$.

The universal properties of Kan extensions assure that the extensions are optimal.

\begin{figure}[h]
\[ (a) \xymatrix{ &  \ar@<20pt>@{=>}[d]^{\gamma'} &\\
 \X \ar[r]^{K} \ar@/_3pc/[rr]_F & \Y \ar@/^3pc/[r]^{H} \ar@{.>}[r]_{\underline{F}_K} \ar@{=>}[d]^{\underline{\psi}}  & \C \\ 
 & & &
  } \qquad \quad \quad
 (b) \xymatrix{ &  \ar@<20pt>@{<=}[d]^{\delta'} &\\
 \X \ar[r]^{K} \ar@/_3pc/[rr]_F & \Y \ar@/^3pc/[r]^{H} \ar@{.>}[r]_{\overline{F}_K} \ar@{<=}[d]^{\overline{\psi}}  & \C \\ 
 & & &
  } \]
\caption{ \\(a) Right Kan Extension is universal \quad (b) Left Kan Extension is couniversal}  
\label{Fig: universal extensions}
\end{figure}
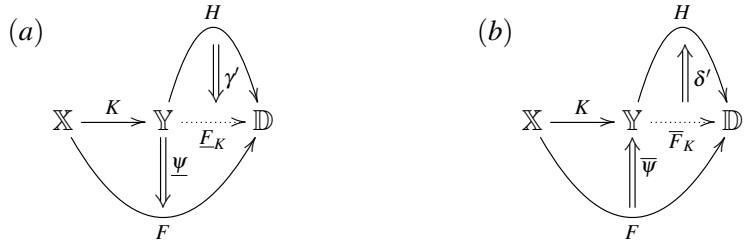

\begin{example}
\label{Eg: Limit example}
The left and right Kan extensions of a functor $F$ along a terminal functor $(!)$ gives precisely the limit and the colimit of $F$. The terminal functor maps all the objects and the maps of the domain category to the single object and the single map in the terminal category ($1$) respectively. Any functor proceeding from the terminal category chooses precisely one object and its identity morphism in the codomain category.

\[ (a)~~~  \xymatrix{ \X \ar@{->}[r]^{!} \ar@/_3pc/[rr]_F & 1 \ar@{.>}[r]^{\underline{F}_!} \ar@{=>}[d]^{\underline{\psi}}  & \C \\ 
 & & &
  }
   ~~~~~~~~~~ 
 (b)~~~  \xymatrix{ \X \ar@{->}[r]^{!} \ar@/_3pc/[rr]_F & 1 
 \ar@{.>}[r]^{\overline{F}_!} \ar@{<=}[d]^{\overline{\psi}}  & \C \\ 
 & & &
  } \]

The {\bf left Kan extension} of $F: \X \to 1$ along the unique functor into $1$ gives a {\bf colimiting cocone}. The functor $\underline{F}_!$ chooses precisely one object in $\C$ (hence we write the object as $\underline{F}_!$) which is the apex of the cocone. The natural transformation $\overline{\psi}$ has components, $\overline{\psi}_X: ! F(X) \Rightarrow \underline{F}_!$ for each $X \in \X$. 

Due to the couniversal property of $\underline{\psi}$, for any other functor $P: 1 \to \C$ with a natural
transformation $\gamma: !P \Rightarrow F$, there
exists a unique natural transformation $\gamma': P
\Rightarrow !\underline{F}_!$ such that $\gamma' \underline{\psi} =
\gamma$. 
Hence, $\underline{F}_!$ is {\bf the limit of
diagram  $F$}. 

\[ 
\xymatrix{
F(A) \ar[dr]^{\overline{\varphi}_A} 
\ar@/_1pc/[ddr]_{\overline{\gamma}_A} 
\ar[rr]^{} & 
  & 
F(B) \ar[ld]_{\overline{\varphi}_B} \ar@/^1pc/[ddl]^{\overline{\gamma}_B}  \\  
& F_! := {\sf colim} F \ar@{.>}[d] & \\ 
& P &
}
\]

Suppose $\C$ is the poset $(\R, \leq)$, then in the above diagram, 
$\underline{F}_!$ is precisely the {\bf greatest lower bound} of $\{ F(A), F(B) \}$. 

Similarly, {\bf the right Kan extension} gives a {\bf limiting cone}. $\overline{F}_!$ is referred to as the {\bf limit of diagram $F$}. When $\C$ is a poset $\overline{F}_!$ is the {\bf least upper bound} of the subset of $\R$ chosen by $F$. 
\end{example}

\subsection{How to compute Kan extensions?}
\label{Sec: Computing Kan extensions}

In Example~\ref{Eg: Limit example}, it was shown that the left and right
Kan extensions of a functor $F$ along the unique functor into the
terminal category are respectively the colimit and the limit  of diagram $F$. 
In this section, we show how one can compute Kan extensions of a functor  when the target category is complete (has all small limits) and cocomplete (has all small colimits) and the intermediate category is locally small (the arrows between any two objects in the category is a small set).

\begin{theorem}\cite[Thoerem 6.2.1]{Rie17}
\label{thm: computing Kan extensions}
Given functors $F: \X \to \C$ and $K: \X \to \Y$, if the category
$\C$ is cocomplete, then the left Kan extension $\overline{F}_K$ 
exists and is defined to be:
\begin{equation}
\label{eqn: left Kan extension}
    \forall Y\in \Y,~~ \overline{F}_{K}(Y) := {\sf colim}( K \downarrow Y \to^{\pi_{K \downarrow Y}} \X \to^{F} \C )
\end{equation}
with the natural transformation $\overline{\psi}$ extracted from  colimiting cocones in $\C$.

If $\D$ is complete, then the right Kan extension $\underline{F}_K$ exists and is defined to be:
\begin{equation}
\label{eqn: right Kan extension}
\forall Y\in \Y,~~ \underline{F}_{K}(Y) := {\sf lim}( Y \downarrow K \to^{\pi_{Y \downarrow K}} \X \to^{F} \C )
\end{equation}
with the natural transformation $\underline{\psi}$ extracted from limiting cones in $\C$.
\end{theorem}
\begin{proof}(Sketch)

Suppose $F: \X \to \C$ is any functor and $\C$ is cocomplete. Then, one can compute the left Kan extension $(\underline{F}_K, \overline{\psi})$ of $F$ along any functor $K: \X \to \Y$ as follows:

\begin{description} 

\item[Defining functor $\underline{F}_K: \Y \to \C$:]~

The left Kan extension is computed on each point (object) in $\Y$. 

For each object $Y$ in $\Y$, consider the slice category $(K \downarrow Y)$. The objects in the slice category are pairs $(X, f)$ where,
\[ f: K(X) \to Y \in \Y \]
and a map $m: (X,f) \to (X, f')$ in the slice category is a map $m \in \X$ such that the following triangle commutes:
\[ \xymatrix{ K(X) \ar[rr]^{K(m)} \ar[dr]_{f} & &  K(X') \ar[ld]^{f'} \\
      & Y &
}\]
Stated informally, the slice category contains complete information on how to arrive at an object $Y \in \Y$ using objects and transformations of $\X$. The projection functor $\pi_{K \downarrow Y}$ chooses precisely the subcategory of $\X$ relevant to $Y$, see Figure~\ref{Fig: object extensions}-(a). The left Kan extension on point $Y$ is the colimit of the diagram $F$ applied to this sub-category. The couniversal cocone of the diagram $\pi_{K \downarrow Y} F$ has a natural transformation $\lambda: {\sf Lim }(\pi_{K \downarrow Y} F) \Rightarrow F$, with a component $\lambda_X$ for each object $\pi_{K \downarrow Y}(f, K(X)) := X \in \X$.

The left extension $\overline{F}_K$ is then defined as follows:
\begin{itemize}
\item For all objects $Y \in \Y$, $\underline{F}(Y) := 
{\sf colim}(\pi_{K \downarrow Y} F)$.

\item  For all maps $h: Y \to Y'$, ${\sf colim}(\pi_{K \downarrow Y} F) \to {\sf colim}(\pi_{K \downarrow Y'} F)$ is the unique arrow induced by $h$, see Figure~\ref{Fig: object extensions}-(b).
\end{itemize}

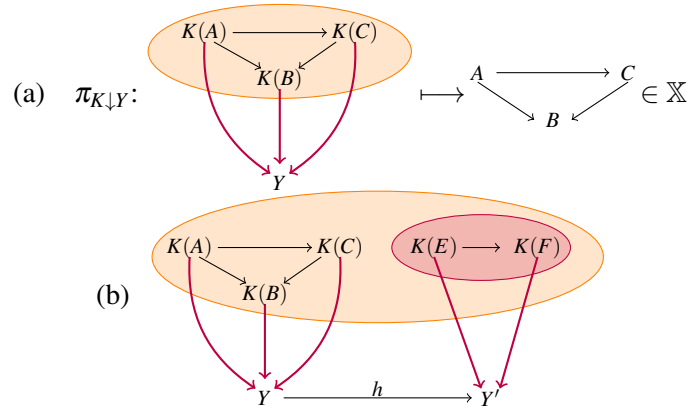
\begin{figure} 
\centering
(a)~~~ $\pi_{K \downarrow Y}$:
\begin{tikzpicture}[scale=1]
	\begin{pgfonlayer}{nodelayer}
		\node [style=none, scale=1.5] (Y) at (-2, 1) {$Y$};
		\node [style=none, scale=1.5] (KA) at (-4, 5) {$K(A)$};
		\node [style=none, scale=1.5] (KB) at (-2, 3.75) {$K(B)$};
		\node [style=none, scale=1.5] (KC) at (0, 5) {$K(C)$};
		\node [style=none] (0) at (-4, 4.75) {};
		\node [style=none] (1) at (0, 4.75) {};
		\node [style=none] (2) at (-2.25, 1.25) {};
		\node [style=none] (3) at (-1.75, 1.25) {};
		\node [style=none] (4) at (-2, 1.5) {};
		\node [style=none] (5) at (-2, 3.5) {};

		\node [style=none] (10) at (-3.75, 4.75) {};
		\node [style=none] (11) at (-0.5, 4.75) {};
		\node [style=none] (15) at (-1.5, 4.05) {};
		\node [style=none] (16) at (-2.5, 4.05) {};
		\node [style=none] (17) at (-3.25, 5) {};
		\node [style=none] (18) at (-0.75, 5) {};
	\end{pgfonlayer}
	\begin{pgfonlayer}{edgelayer}
		\draw [->] (17.center) to (18.center);
		\draw [->] (11.center) to (15.center);
		\draw [->] (10.center) to (16.center);
		\draw[->, color=purple, thick] [bend right] (0.center) to (2.center);
		\draw[->, color=purple, thick] [bend left] (1.center) to (3.center);
		\draw[->, color=purple, thick]  (5.center) to (4.center);
	\end{pgfonlayer}
	\begin{pgfonlayer}{backgroundlayer}
		\node [ellipse, draw, minimum height = 2.5cm, minimum width=7cm, fill=orange, fill opacity=0.2, color=orange] (e) at (-2, 4.5) {};
	\end{pgfonlayer}
\end{tikzpicture} $\longmapsto$ \begin{tikzpicture}[scale=1]
	\begin{pgfonlayer}{nodelayer}
		\node [style=none, scale=1.5] (A) at (-4, 5) {$A$};
		\node [style=none, scale=1.5] (B) at (-2, 3.75) {$B$};
		\node [style=none, scale=1.5] (C) at (0, 5) {$C$};
		\node [style=none] (0) at (-4, 4.75) {};
		\node [style=none] (1) at (0, 4.75) {};
		\node [style=none] (5) at (-1.5, 3.75) {};
		\node [style=none] (6) at (-2.5, 3.75) {};
		\node [style=none] (7) at (-3.5, 5) {};
		\node [style=none] (8) at (-0.5, 5) {};
	\end{pgfonlayer}
	\begin{pgfonlayer}{edgelayer}
		\draw [->] (7.center) to (8.center);
		\draw [->] (1.center) to (5.center);
		\draw [->] (0.center) to (6.center);
	\end{pgfonlayer}
\end{tikzpicture} $\in \X$
~~~~~~~~(b)~~ \begin{tikzpicture}[scale=1]
	\begin{pgfonlayer}{nodelayer}
		\node [style=none, scale=1.5] (Y) at (-2, 1) {$Y$};
		\node [style=none, scale=1.5] (KA) at (-4, 5) {$K(A)$};
		\node [style=none, scale=1.5] (KB) at (-2, 3.75) {$K(B)$};
		\node [style=none, scale=1.5] (KC) at (0, 5) {$K(C)$};
		\node [style=none] (0) at (-4, 4.75) {};
		\node [style=none] (1) at (0, 4.75) {};
		\node [style=none] (2) at (-2.25, 1.25) {};
		\node [style=none] (3) at (-1.75, 1.25) {};
		\node [style=none] (4) at (-2, 1.5) {};
		\node [style=none] (5) at (-2, 3.5) {};

		\node [style=none] (10) at (-3.75, 4.75) {};
		\node [style=none] (11) at (-0.5, 4.75) {};
		\node [style=none] (15) at (-1.5, 4.05) {};
		\node [style=none] (16) at (-2.5, 4.05) {};
		\node [style=none] (17) at (-3.25, 5) {};
		\node [style=none] (18) at (-0.75, 5) {};
		
		\node [style=none] (15) at (-1.5, 4.05) {};
		\node [style=none] (16) at (-2.5, 4.05) {};
		\node [style=none] (17) at (-3.25, 5) {};
		\node [style=none] (18) at (-0.75, 5) {};
		\node [style=none, scale=1.5] (19) at (4, 1) {$Y'$};
		\node [style=none, scale=1.5] (20) at (2.5, 5) {$K(E)$};
		\node [style=none] (21) at (5.25, 5) {};
		\node [style=none, scale=1.5] (22) at (5.25, 5) {$K(F)$};
		\node [style=none] (23) at (3.25, 5) {};
		\node [style=none] (24) at (4.25, 5) {};
		\node [style=none] (25) at (3.75, 1.25) {};
		\node [style=none] (26) at (4.25, 1.25) {};
		\node [style=none] (27) at (2.5, 4.75) {};
		\node [style=none] (28) at (5.25, 4.75) {};
		\node [style=none] (e2) at (3.75, 5) {};
		\node [style=none] (29) at (-1.5, 1) {};
		\node [style=none] (30) at (3.5, 1) {};
	\end{pgfonlayer}
	\begin{pgfonlayer}{edgelayer}
		\draw [->] (17.center) to (18.center);
		\draw [->] (11.center) to (15.center);
		\draw [->] (10.center) to (16.center);
		\draw[->, color=purple, thick] [bend right] (0.center) to (2.center);
		\draw[->, color=purple, thick] [bend left] (1.center) to (3.center);
		\draw[->, color=purple, thick]  (5.center) to (4.center);
		\draw [->] (23.center) to (24.center);
		\draw [->] (29.center) to (30.center);
		\draw [->, color=purple, thick] (27.center) to (25.center);
		\draw [->, color=purple, thick] (28.center) to (26.center);
		\node [style=none, scale=1.5] (31) at (1, 1.25) {$h$};
	\end{pgfonlayer}
	\begin{pgfonlayer}{backgroundlayer}
		\node [ellipse, draw, minimum height = 3.5cm, minimum width=12cm, fill=orange, fill opacity=0.2, color=orange] (e) at (1, 4.75) {};
		\node [ellipse, draw, minimum height = 1.75cm, minimum width=4.75cm, color=purple, fill=purple, fill opacity=0.2] (e2) at (3.75,5) {};
	\end{pgfonlayer}
\end{tikzpicture}
\caption{ (a) $\pi_{K \downarrow Y}$ projects the shaded region of $(K \downarrow Y)$ into $\X$ (in general $K \downarrow Y$ is not a subcategory of $\X$); (b) An arrow $h: Y \rightarrow Y' \in \Y$ leads to the (shaded) base of above $Y'$ to be included in the (shaded) base above $Y$. Hence the base of the colimiting cocone of $\pi_{K \downarrow Y}F$ includes the shaded base of the colimiting cone of $\pi_{K \downarrow Y}F$ inducing a unique map ${\sf colim}(\pi_{K \downarrow Y}F) \rightarrow {\sf colim}(\pi_{K \downarrow Y'}F)$}
\label{Fig: object extensions}
\end{figure}

\item[Defining the natural transformation $\overline{\psi}: F \Rightarrow K \overline{F}_K$ :]~

For all $X \in \X$, $\overline{\psi}_X$ is the component ${\sf lim}(\pi_{K \downarrow KX}F) \to F(X)$ of the colimiting cocone corresponding to the initial object $(1_KX, KX) \in (K \downarrow KX)$.


\end{description}

Computing right Kan extension is dual to computing left Kan extensions. 
If $F: \X \to \C$ is any functor and $\C$ is complete (contains all small limits), then one can compute the right Kan 
extension $\underline{F}_K, \underline{\psi})$ 
of $F$ along any functor $K: \X \to \Y$ 
by replacing the slice construction by the coslice construction, and colimits by limits in the above procedure.
\end{proof}

\begin{corollary} 
\label{Lem: full and faithful}
 If $(\underline{F}_K, \underline{\psi})$ is the right Kan extension of a functor $F$ along any {\bf full and faithful} functor $K$, then the natural transformation $\underline{\psi}$ is an isomorphism.
 
 Similarly, if $(\overline{F}_K, \overline{\psi})$ is the left Kan extension of a functor $F$ along any {\bf full and faithful} functor $K$, then the natural transformation $\overline{\psi}$ is an isomorphism.
\end{corollary}

\begin{proof}

Note that for all $X \in \X$, $ K\underline{F}(X) = \underline{F}(K(X)) :=  {\sf lim} (\pi_{K(X) \downarrow K}) $

Since $K$ is full and faithful, every $K(f): K(X) \to K(X') \in \Y$ corresponds to a unique $f: X \to X' \in \X$. Then for all $X \in \X$, $(KX, 1_{KX})$ is an initial object in the coslice category $(KX \downarrow K)$. Thereby, the diagram $\pi_{KX \downarrow K}$ contains all the maps radiating from $X$. Hence, ${\sf lim} (\pi_{KX \downarrow K}F) = F(X)$, thereby, $\underline{\psi}$ is an isomorphism. 

The argument for the left Kan extension is dual to the above proof.
\end{proof}

We use this procedure to compute extensions of resource monotones which  are functors into a posetal category (a poset considered as a category), see Section~\ref{Sec: monotone extensions}.

\section{Kan extensions of Resource Measures}

\subsection{Resource Theories as partitioned Categories}
\label{Sec: categorical framework}

We introduce partitioned Categories as a framework for resource theories, and functors for partitioned Categories to describe relationships between resource theories. 

\begin{definition} 
A {\bf partitioned category (pCat)} $(\X, \X_\f)$ consists of a category $\X$ and a chosen subcategory $\X_\f$ of free transformations with the inclusion being bijective on objects. 
\end{definition}

The objects of the category are interpreted as {\bf resources} and the maps to be
{\bf resource transformations}. The subcategory includes all objects and those 
transformations which are designated to be {\bf free}. 

The following are a few examples of resource theories as pCats:
\subsubsection*{Randomness}

Cryptographic protocols use randomness as an essential resource for establishing secure communication of devices by generating random keys. 
The degree of randomness determines how secure the communication channel is. 
Randomness is also used in computer algorithms to solve certain problems. In other words, randomness is an essential computational resource of practical use. 
Entropy is used as measure of randomness: in particular, Shannon entropy quantifies randomness in that it expresses the average surprisal on the outcome of a random experiment. Entropy has been studied in the context of randomness using the category ${\sf FinProb}$ (renamed below as ${\sf Detmn}$)  \cite{BFL11} and \cite[Example 2.5]{CFS16}. The following is a resource theory of randomness:

\begin{example}
\label{eg: RandDetmn}
{$\mathsf{(Rand, Detmn)}$:} 

($\mathsf{Detmn}$ is the chosen sub-category of free transformations in $\Rand$)

\begin{description} 
\item[Resources:] $(X,p)$ where $X$ is a finite set and $p$ is a probability distribution over $X$. 

$X$ can be interpreted of as a set of possible states of a system and $p$ be the probability distribution over the states.
\item[Resource Transformations:] $M: (X, p) \to (Y, q)$ is a real $|X| \times |Y|$ row stochastic matrix (rows sum to 1) such that $pM=q$. 

A resource transformation $M: (X, p) \to (Y, q)$ is row stochastic if and only if for all $x \in X$, $M_x$ is a probability distribution: suppose the system is in state x, then the stochastic process produces states $y\in Y$ with probability $M_{xy}$.The requirement that $pM = q$ means that under the stochastic process $M$, the probability of $Y$ being in state $y$ after process $M$ on $X$ is given by $\sum_{x \in X} M_{xy} p_x$.

\item[Identity transformations:] Identity matrices
\item[Composition:] Suppose $(X,p) \xrightarrow{M} (Y,q) \xrightarrow{N} (Z,s)$, then $(X,p) \xrightarrow{MN} (Z,s)$ is defined as the matrix multiplication
\item[Free transformations:] A resource transformation $(X,p) \xrightarrow{M} (Y,q)$ is free if it is deterministic, that is, $M$ is simply a function $X \to Y$. Hence, for all $x \in X, y \in Y$, $M_{xy} \in \{0, 1\}$
\end{description}
\end{example}

\subsubsection*{Non-uniformity}
{\em Pure states} represent states on which the experimenter has maximum information. These conditions are often very hard to achieve in concrete settings due to the presence of external noise. In such cases, the state is called mixed, and can be expressed as a convex combination of pure states. 
From this perspective, it is clear that pure states represent the maximal resource and the closer a state is to a pure state, the more resourceful it is. Therefore, the least resourceful state of any system is the \emph{maximally mixed state}, which can be expressed as a uniform probability distribution over the states of the system.
\cite{GMV15} 

\begin{example}
\label{eg: RandUni}
{$\mathsf{(Rand, Uniform)}$:}
\begin{description} 
\item[Resources, transformations, identity and composition:] Same as example~\ref{eg: RandDetmn}
\item[Free transformation:] A map $(X,p) \xrightarrow{U} (Y,q)$ is free if $U$ is a uniform matrix. A row stochastic matrix $(X,p) \to^{M} (Y,q)$ is uniform if for all $y \in Y$,
\[ \sum_{x \in X} M_{xy} = 1 \]
The columns of $M$ sum to $|X| / |Y|$. When $U$ is a square matrix, it is doubly stochastic.
\end{description}
\end{example}

Note that, a uniform probability distribution $u := (1/n, 1/n, 1/n, \cdots, 1/n)$ is simply the uniform matrix $(\{*\}, (1)) \to (X,u)$, which is $u$ itself.

$({\sf Rand, Uniform})$ consists of classical probabilistic states. 
A non-uniformity theory based on quantum states is as follows:

\begin{example}$(\qRand, \qUniform)$
\label{eg: qRandUni}
\begin{description}
\item[Resources:] $(\rho, H)$ where $\rho: H \to H \in L(H)$ is a quantum state, also known as density matrix (a positive semi-definite operator with trace $1$), and $H$ is a finite-dimensional Hilbert space. 
\item[Resource transformations:] $(\rho, H) \to^{\mathcal{E}} (\sigma, K)$ is a quantum channel $\mathcal{E}: L(H) \to L(K)$ such that 

$\mathcal{E}(\rho) = \sigma$
\item[Composition and Identity transformations:] Usual composition of quantum channels and identity channels
\item[Free transformations:] Unital quantum channels i.e., $\chan{E}: L(H) \to L(K)$ such that $\chan{E}\left(\maxmix{H} \right) =  \maxmix{K}$, where $\mathds{1}_H \in L(H)$ is the identity matrix. In other words, unital channels preserve maximally mixed states.
\end{description}
\end{example}

\subsubsection*{Entanglement}

Entanglement is one of the most important quantum resources, and it is used in several communication scenarios, such as quantum teleportation \cite{Teleportation} or dense coding \cite{Dense-coding}. It is known that local operations and classical communication (LOCC) cannot increase the entanglement of a quantum state~\cite{Review-entanglement}. Hence, when entanglement is considered to be a resource, LOCC operations  are precisely the free transformations of this resource theory. The basic setting in which entanglement is studied involves quantum states over two systems, which is referred to as ``bipartite entanglement''.

A resource theory of bipartite entanglement is constructed as follows. The following resource theory is obtained by applying the coslice (state) construction on \cite[Example 3.7]{CFS16}:

\begin{example} \label{eg: Bip} ${\sf (Bip, LOCC) }$:
\begin{description} 
\item[Resources:] $\rho \in L(H \ox K)$ is a quantum state which is a positive semi-definite operator with trace $1$, and $H$, $K$ are finite-dimensional Hilbert spaces. 
\item[Resource transformations:]
$\rho \to^{\mathcal{E}} \sigma$ is a quantum channel (completely positive trace preserving map) such that $\mathcal{E}(\rho) = \sigma$

\item[Free transformations:] Local operations and classical communication
\end{description}
The composition is the usual composition of identity channels.
\end{example}


\subsubsection*{Distinguishability} 
In some situations it is important to consider pairs of quantum states and evaluate how different they are from each other. To this end, various quantifiers have been defined, such as the trace distance, the fidelity \cite{Wilde} or quantum divergences \cite{GoT20,Gou21}. These quantifiers all show that, whenever the same channel is applied to each element of a pair of quantum states, in general our ability to distinguish the resulting states is decreased. This suggests setting up a resource theory of the distinguishability, also known as quantum relative majorization \cite{rel-sub-maj,quant-rel-lor-curv}.

A resource theory of quantum distinguishability is given as follows \cite{GoT20,Gou21}.

\begin{example}
\label{eg: Divergence}{\sf (Distinguish, Processing)}:

\begin{description} 
\item[Resources]: $((\rho,\sigma), H)$ are pairs of quantum states, that is, $\rho , \sigma \in L(H)$ where $H$  is a finite-dimensional Hilbert space.

\item[Resource transformations]: $(\mathcal{E}_1, \mathcal{E}_2):
((\rho_H, \sigma_H),H) \to ((\rho_K, \sigma_K), K)$ are pairs of
quantum channels $\mathcal{E}_1, \mathcal{E}_2: L(H) \to L(K)$ such
that $\mathcal{E}_1(\rho_H) = \rho_K$ and $\mathcal{E}_2(\sigma_H) =
\sigma_K$

\item[Composition and identity transformations]: $(\mathcal{E}_1,
\mathcal{E}_2)(\mathcal{E}_3, \mathcal{E}_4) := (\mathcal{E}_1
\mathcal{E}_3, \mathcal{E}_2 \mathcal{E}_4)$ and identity
transformations are given by pairs of identity channels
\item[Free transformations:] $(\mathcal{E}_1, \mathcal{E}_2)$ such that $\mathcal{E}_1 = \mathcal{E}_2$
\end{description}
\end{example}

\subsection{Relationships between Resource Theories as pCat functors}

Now that we have formalized resource theories as pCats, we can formalize the relationship between resource theories as functors of pCats. For example, classical theories of the corresponding quantum resource theories. Physical theories defined on pure states are considered as subtheories of corresponding mixed state theories. Such relationships  can be formalized as functors of pCats. 

\begin{definition}
A {\bf functor of partitioned categories (pCat)}, $F: (\X,\X_\f) \to (\Y, \Y_\f)$, is a functor $F: \X \to \Y$ such that if $f \in \X_\f$ then $F(f) \in \Y_\f$ i.e., the functor preserves free transformations.
\end{definition}

$F: (\X, \X_{\sf f})$ being a functor means that it preserves the identity transformations: $F(1_A) = 1_{F(A)}$, and it preserves the composition in $\X$: $F(fg) = F(f) F(g)$.  

 Figure~\ref{Fig: pSMC functor} is a schematic of a pCat functor. The triangles represent non-free transformations, and the hollow circles represent free transformations. As one can see, a pCat functor may or may not preserve a non-free transformation. 

\begin{figure}[h]
 \centering
 \includegraphics[width=0.4\textwidth]{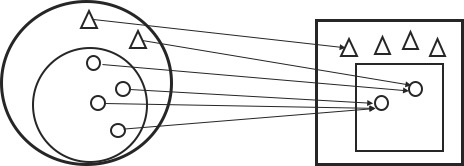} 
 \caption{Schematic for functor of pCats}
 \label{Fig: pSMC functor}
\end{figure}

\begin{definition}
A pCat functor $F:(\X,\X_\f) \to (\Y, \Y_\f)$ is {\bf full} if $F: \X \to \Y$ is full, and $F$ is {\bf faithful} if $F: \X \to \Y$ is faithful.
\end{definition}

Let us look at a few examples of full and faithful pCat functors. For a quantum system, pure states are considered to be a subset of mixed states since mixed states are convex combination of pure states. 
The resource theory of pure states bipartite entanglement embeds in the general theory of bipartitle entanglement through an inclusion functor, see Example~\ref{eg: Bip inclusion}. Lemma~\ref{lem: Bip inclusion} proves that this inclusion is a pCat functor.

\begin{example}
\label{eg: Bip inclusion}
The resource theory of bipartite pure-state  entanglement, ${\sf (PureBip, LOCC_p)}$ has pure quantum states 
$\rho$, where $\rho \in L(H \ox K)$, as resources. A resource transformation is a
quantum channel $\mathcal{E}: \rho \to \sigma$ such that $\mathcal{E}(\rho) = \sigma$ where
$\sigma \in L(H' \ox K')$. The free transformations are LOCC operations between
pure bipartite states, here denoted as ${\sf LOCC_p}$.
\end{example}

\begin{lemma} 
\label{lem: Bip inclusion}
The inclusion $i: {\sf PureBip} \hookrightarrow  {\sf Bip}$ defined to be identity on objects and maps is a full and faithful pCat functor $i: {\sf (PureBip, LOCC_p)} \hookrightarrow ({\sf Bip, LOCC})$. 
\end{lemma}
\begin{proof} 
The inclusion is a pCat functor since ${\sf LOCC_p} \hookrightarrow {\sf LOCC}$. Moreover, $i: {\sf PureBip} \hookrightarrow  {\sf Bip}$ is full and faithful inclusion.
\end{proof}

Classical theories are considered as sub-theories of quantum theories. This gives an inclusion functor classical distinguishability into quantum distinguishability. The following is the resource theory for classical distinguishability and is referred to as classical relative majorization in \cite{mix-cha, mix-dis, gen-har,rel-sub-maj,quant-rel-lor-curv}:

\begin{example}
\label{eg: cDistinsguish}
In the resource theory of classical distinguishability, ${\sf (cDistinguish, cProcessing)}$, a
resource $((p,q), X)$ is a pair of probability distributions $p:= (p_1, \cdots, p_{|X|})$ and $q := (q_1, \cdots, q_{|X|})$ over a finite set $X$. Resource
transformations $(M, M'):((p,q), X) \to ((p',q'), Y)$ where $M$ and $M'$ are pairs of row stochastic matrices such that $pM = p'$ and $q M' = q'$. Free transformations are $(M, M')$ such that $M = M'$. 
\end{example}

\begin{example} 
\label{eg: Div inclusion}
The  inclusion $i: {\cDivergence} \hookrightarrow {\Divergence}$ is defined as follows: 

\begin{itemize} 
\item For all resources $((p,q), X) \in {\cDivergence}$, $i((p,q), X) := ((\rho^p,\rho^q), \D^{|X|})$ where $[\rho^p]_{ij} = \delta_{ij} p_i$ , $1 \leq i \leq |X|, 1 \leq j \leq |X|$. $\rho^p$ and $\rho^q$ are diagonal density matrices with the probability distributions $p$ and $q$ as their diagonals respectively.

\item Given a transformation $(M, M'): ((p,q), X) \to ((p', q'),Y)$, $i((M, M')) := (\mathcal{E}, \mathcal{E}')$ where $\mathcal{E}$ and $\mathcal{E'}$ are determined by $M$ and $M'$ respectively as follows.

For any quantum state (positive semi-definite operator of trace 1 on a Hilbert Space $H$), 

\begin{equation}
\label{eqn: stochastic to quantum channel}
\mathcal{E}(\rho)=\sum_{i,j}B_{ij}\rho B_{ij}^{\dagger}
\end{equation}

where $B_{ij}=\sqrt{M_{ij}}|j\rangle\langle i|$ where $1 \leq i, j \leq |X|$ and $B_{ij}^\dagger$ is its adjoint (cf.\ \cite{Wilde}). $(|j\rangle$ is a column vector with $1$ at position $j$ and zero elsewhere.$)$
\end{itemize}
\end{example}

\begin{lemma} 
\label{Lemma: Lemma: Div inclusion ff}
The inclusion $i: {\cDivergence} \hookrightarrow {\Divergence}$ defined in Example~\ref{eg: Div inclusion} is full and faithful (or fully faithful) pCat functor $i: { (\cDivergence, \cProcessing)} \hookrightarrow (\Divergence, \Processing)$. 
\end{lemma}

\begin{example}
\label{eg: non-uniformity inclusion}
Closely, related to Example~\ref{eg: Div inclusion}, is the inclusion of $\Rand$ into $\qRand$. The inclusion $i: \Rand \hookrightarrow \qRand$ is defined as follows: for all $(p, X) \in \Rand$, $i((p, X)) := (\rho^p, \mathbb{C}^{|X|})$, and for row stochastic matrices $M \in \Rand$, $i(M)$ is defined as in eqn~\eqref{eqn: stochastic to quantum channel}.
\end{example}

\begin{lemma}
\label{Lemma: non-uniformity inclusion}
The  inclusion $i: \Rand \hookrightarrow \qRand$ defined in Example~\ref{eg: non-uniformity inclusion} is a full and faithful pCat functor $i: { (\Rand, \Uniform)} \hookrightarrow (\qRand, \qUniform)$. 
\end{lemma}

\subsection{Preorder collapse and monotones}

One of the major goals of resource theories is to the identify necessary and sufficient conditions for the existence of a free transformation between two resources. Once such conditions are identified, one can choose to `forget' the different possible ways in which resource $A$ can be converted to resource $B$ freely, and only `remember' if there exists a free transformation from $A$ to $B$. 

The necessary and sufficient conditions for the existence of a free transformation between pairs of resources define an equivalence class on the resource theory (freely inter-convertible resources are considered to be equivalent) and a preorder on the equivalence classes. Such necessary and sufficient conditions can be encoded into a pCat functor from the resource theory. On applying this functor, the resource theory collapses into a preorder:

\begin{definition} Given a resource theory $(\X, \X_\f)$ and a preorder  $({\sf ob}(\X), {\sf order})$ where ${\sf ob}(\X)$ refers to the set of objects of $\X$, a {\bf preorder collapse} of the resource theory $(\X, \X_\f)$ is a pCat functor $(\X, \X_\f) \to (\chaos_\X, \order_{\X})$, where $\chaos_\X$ is the indiscrete (chaotic) category with the same objects as $\X$, and for any two objects $A, B \in \chaos_\X$, the transformation $A \to B \in \order_\X$ if ``$A ~{\sf \order}~ B$" is true. 
\end{definition}

Let us look at an example of a preorder collapse of ${\sf (Rand^\op, Uniform)}$ determined by the majorization relation \cite{MOA11}. 
Suppose $p:= (p_1, p_2, p_3, \cdots, p_n)^\uparrow$ and $q:= (q_1, q_2, q_3, \cdots, q_m) ^\uparrow$ such that the elements of the distribution are in increasing order.
We say $q$ is majorized by $p $ written as $q \preceq p$ if the Lorenz curve \cite{MOA11, Lor05} of $p$ lies either completely below the Lorenz curve of $q$ (see Figure~\ref{Fig: Lorenz}) or coincides with it. This means that $q$ is more uniform than $p$. 

 {\bf Lorenz curve} \cite{Lor05, MOA11, GMV15} $L(p)$ for a probability distribution $p := (p_1, p_2, \cdots, p_n)$ is characterized as the linear interpolation of points 
 $(i/n, \sum_{k=1}^{i} p_k)$, where $i=0,1, \cdots, n$; see Figure~\ref{Fig: Lorenz}.
 
 \begin{figure} 
 \centering 
 \includegraphics[scale=0.5]{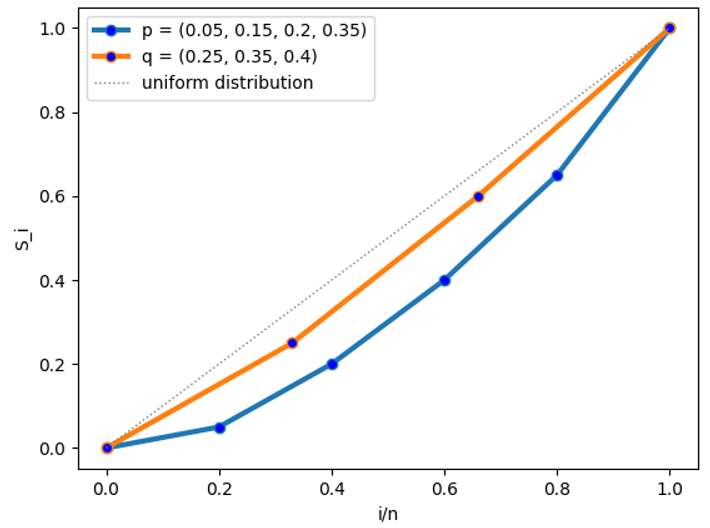}
 \caption{The Lorenz curve of $q$ is majorized by the Lorenz curve of $p$}
 \label{Fig: Lorenz}
 \end{figure}

\begin{example}
\label{defn: majorization}
Define $P: {\sf Rand} \to {\sf chaos_{Rand}}$ as follows: for each probability distribution $p \in {\sf Rand}$, $P(p) := p$; each $M: p \to q$ is mapped to the unique arrow $p \to  q$.
\end{example}

The functors into chaotic categories are determined by the objects. It is straightforward that $P$ as defined above is a functor. The below theorem establishes that $P: ({\sf Rand^\op, Uniform}) \rightarrow$ $(\chaos_{\sf Rand}, \preceq_{\sf Rand})$ is pCat functor.

\begin{theorem}\cite{HLP29}
\label{Thm: Majorization}
Given two finite probability distributions $p$ and $q$, $q \preceq p$ if and only if there exists a uniform matrix $U: p \to q$ such that $pU = q$.
\end{theorem}

\begin{corollary} 
The map $P: ({\sf Rand^\op, Uniform}) \rightarrow (\chaos_{\sf Rand}, \preceq)$ is a preorder collapse.
\end{corollary}
\begin{proof} 
By Theorem~\ref{Thm: Majorization}, if $U: p \to  q \in {\sf Uniform}$, then $P(q) \preceq P(p)$.
\end{proof}

Given a resource theory, one can assign a real value to each resource such that the assignment respects the preorder defined by the theory. In this respect, a monotone is a function $f: R \to [0, \infty]$ where $R$ is a set of resources which preserves the preorder on the equivalence class of resources. To represent a monotone as a pCat functor, the poset $([0,\infty], \leq)$ is defined as a pCat as follows:- 

\begin{definition}
\label{defn: pCat poset}
The poset $([0,\infty], {\leq})$ is encoded as the  pCat $(\chaos_{[0, \infty]}, \leq_{[0,\infty]})$ where $\chaos_{[0, \infty]}$ is the chaotic category with objects as $r \in [0, \infty]$ and the free transformations are those maps respecting the $\leq$ order $(m \to n \in \leq_{[0,\infty]}$ if and only if $m \leq n)$.
\end{definition}

\begin{definition}
\label{defn: monotone}
A {\bf monotone} for a resource theory $(\X, \X_\f)$ is a pCat  functor  \[ F: (\X, \X_\f) \to \left(\chaos_{[0,\infty]}, \leq_{[0, \infty]} \right).\] 

An {\bf op-monotone} is a contravariant monotone, that is,  \[F: (\X, \X_\f) \to \left(\chaos_{[0,\infty]}, \leq_{[0, \infty]} \right)^\op :=  \left(\chaos_{[0,\infty]}^\op, \geq_{[0, \infty]} \right).\]
\end{definition}

Even though op-monotones are more frequently used in resource theories, we defined the codomain of a monotone to be $(\chaos_{[0,\infty]}, \leq_{[0, \infty]})$ because, in general, an arrow $a \to b$ in a posetal category refers to $a \leq b$, and such ordering becomes relevant when one computes $\inf$ and $\sup$ of a subset in the poset, see Section~\ref{Sec: monotone extensions}. Let us look at a few examples of monotones: 

In information theory, Shannon entropy is a well-known measure of randomness or uncertainty in the outcome when a random experiment (experiment with multiple outcomes) is repeated one or more times. The value of Shannon entropy lies in [0,1] where 0 represents absolute certainty and 1 represents maximum uncertainty. In the following example, we construct a monotone for the resource theory of randomness, ${\sf (Rand, Detmn)}$, and an op-monotone ${\sf (Rand, Uniform)}$ based on the Shannon entropy:

\begin{example}
\label{eg: Shannon entropy}
Define ${\sf Shannon}: {\sf Rand} \to \chaos_{[0, \infty]}$ as follows:
\begin{itemize}
\item For all finite probability distributions, $(X, p) \in \mathsf{Rand}$, ${\sf Shannon}(p) := H(p)$ where $H(p)$ is the Shannon entropy of $p$:
\[H(p) := - \sum_{1 \leq i \leq |X|} p_i ~{\sf log}~p_i \]

\item For all $(X, p) \xrightarrow{{\sf Shannon}} (X, q) \in \mathsf{Rand}$, then $F(M)$ is the unique arrow $H(p) \to H(q)$
\end{itemize}
\end{example}

It is straightforward that ${\sf Shannon}: {\sf Rand} \to \chaos_{[0, \infty]}$ is a functor. We note that the functor 
$\Shannon$ acts as a monotone for $(\Rand, {\sf Detmn})$ and 
as an op-monotone for $(\Rand, \Uniform)$.

\begin{lemma}\cite{CoT06}
\label{Lemma: Shannon}
Suppose ${\sf Shannon}: (X,p) \to (Y,q) \in {\sf Detmn}$, then $H(p) \geq H(q)$.
\end{lemma}

\begin{corollary} 
The map $\Shannon: ({\sf Rand, Detmn}) \rightarrow ([0,\infty], \geq)$ defined as in Example~\ref{eg: Shannon entropy} is an op-monotone.
\end{corollary}

\begin{lemma}\cite{GMV15, MOA11}
\label{Lemma: opShannon}
Suppose ${\sf Shannon}: (X,p) \to (Y,q) \in {\sf Uniform}$, then $H(p) \leq H(q)$.
\end{lemma}

\begin{corollary} 
The map $\Shannon: ({\sf Rand, Uniform}) \rightarrow ([0,\infty], \leq)$ defined as in Example~\ref{eg: Shannon entropy} is a monotone.
\end{corollary}

Next we describe a monotone for the resource theory, ${(\cDivergence, \cProcessing)}$:

\begin{example} \label{defn: divergence monotone} \cite[Definition 2]{GoT20,Gou21}
A {\bf classical divergence} $D: {\cDivergence} \to \chaos_{[0,\infty]}$ is any  functor, that for any resource $((p,q),X) \in {\cDivergence}$ and a resource transformation $(M, M) \in {\cDivergence}$, satisfies the data processing inequality: \[D(p,q) \geq D(pM, qM) \] 
\end{example}

\begin{lemma}
Any classical divergence $D: ({\cDivergence}, {\cProcessing} \to (\chaos_{[0,\infty]}, \geq_{[0,\infty]})$ is an op-monotone.
\end{lemma}

The following is a monotone for ${\sf PureBip}$ as follows \cite{NiC11, Wilde}: 

\begin{example} 
\label{eg: Schmidt}
Define ${\sf Schmidt}: {\sf PureBip} \rightarrow \chaos_{[0,\infty]}$ to be the following: for all resources $\rho^{HK} \in {\sf PureBip}$, where $\rho^{HK} \in L(H \ox K)$ is a pure quantum state, 
\[ N \left(\rho^{HK} \right) := {\rm Rank} \left(\rho^H \right) \text{ where, } \rho^H := {\rm Tr}_K\left(\rho^{HK} \right)\]
\end{example}

\begin{lemma}\cite{NiC11, Wilde}
$N: ({\sf PureBip}, {\sf LOCC_p}) \rightarrow (\chaos_{[0,\infty]}, \geq_{[0, \infty]})$ is an op-monotone.
\end{lemma}



  


\subsection{Kan Extensions of monotones}
\label{Sec: monotone extensions}

Now, we set up resource theories to apply Kan extensions for extending monotones from one resource theory to another when there exists a pCat functor between them. 

Given a monotone $M: (\X, \X_\f) \to (\chaos_{[0, \infty]}, \leq_{[0,\infty]})$ and a pCat functor $K:(\X, \X_\f) \to (\Y, \Y_\f)$, one could desire to extend $M$ to obtain monotones on $(\Y, \Y_\f)$. Observe that a monotone $M: (\X, \X_\f) \to \left( \chaos_{[0, \infty]}, \leq_{[0,\infty]} \right)$ is concerned only with the free transformations: if $f: A \to B \in \X_\f$, then $M(A) \leq M(B)$; otherwise $M(A) \to M(B)$ is the unique arrow signifying that there is no order between $A$ and $B$. Since $\X_\f$ includes all the objects of $\X$, and the monotone $M$ is concerned with only free transformations, in order to extend $M: (\X, \X_\f) \to \pCat{[0, \infty]}{\leq}$ it suffices to extend $M_\f: 
\X_\f \to \leq_{[0, \infty]}$ which is defined as follows :-
\[ M_\f : \X_\f \to \leq_{[0, \infty]};  ~~~ A \to^{f} B \mapsto M(A) \leq M(B) \]

\begin{definition}
\label{defn: monotone extensions}
Let $(\X, \X_\f) \to^{K} (\Y, \Y_\f)$ be a pCat functor. Let $M: (\X, \X_\f) \to (\chaos_{[0,\infty]}, \leq_{[0,\infty]})$ be a monotone for the resource theory $(\X, \X_\f)$. 

\begin{enumerate}[(a)]
\item The {\bf  minimal extension}\footnote{We follow the naming convention in \cite{GoT20} for monotone extensions.} $\underline{M}_K: \Y_\f \to \leq_{[0,\infty]}$ of $M$ along $K$ is the right Kan extension of the functor $M_\f: \X_{\f} \to \leq_{[0, \infty]}$ along the functor $K_\f: \X_{\f} \to \Y_\f; K_\f(h) := K(h)$ (see Figure~\ref{Fig: extensions}-(a)).

\item The {\bf  maximal extension} $\overline{M}_K: \Y_\f \to \leq_{[0,\infty]}$ of $M$ along $K$ is the left Kan extension of the functor $M_\f: \X_{f} \to  \leq_{[0, \infty]}$ along the functor $K_\f: \X_{\f} \to \Y_\f; K_\f(h) := K(h)$ (See Figure~\ref{Fig: extensions}-(b)).  
\end{enumerate}
\end{definition}

\begin{figure}[h]
\[ (a)~~~  \xymatrixcolsep{4pc} \xymatrix{ 
&  \ar@<35pt>@{=>}[d]^{\leq} &\\
\X_{\f} \ar[r]^{K_\f} \ar@/_3pc/[rr]_{M_\f} & \Y_\f \ar@/^3pc/[r]^{G}
\ar@{.>}[r]_---{\underline{M}_K} \ar@{=>}[d]^{\leq}  & \leq_{[0,\infty]} \\ 
 & & &
  }
   ~~~~~ 
 (b)~~~  \xymatrix{ 
 &  \ar@<35pt>@{<=}[d]^{\leq} &\\
 \X_{\f} \ar[r]^{K_\f} \ar@/_3pc/[rr]_{M_\f} & \Y_\f \ar@/^3pc/[r]^{G} \ar@{.>}[r]_---{\overline{M}_K} \ar@{<=}[d]^{\leq}  & \leq_{[0,\infty]} \\ 
 & & &
  } \]
 \caption{(a) Minimal (right Kan) extension \quad \qquad \qquad \quad (b) Maximal (left Kan) extension}  
 \label{Fig: extensions}
\end{figure}
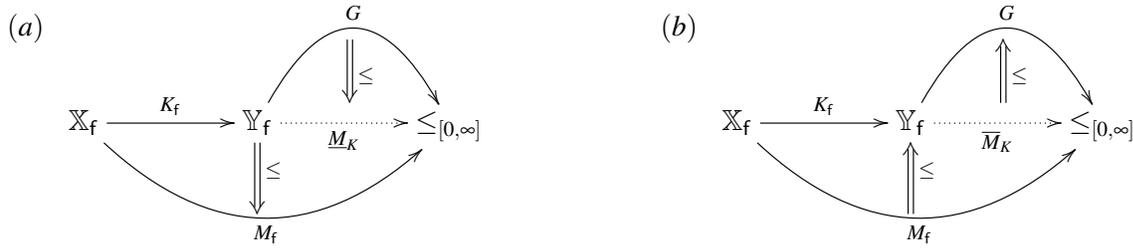

Let us unpack the definition of minimal and maximal extensions of a monotone in Definition~\ref{defn: monotone extensions}. Any category given by a poset with suprema and infima is both complete and cocomplete. Since $([0,\infty], \leq)$ is such a poset,  by Theorem~\ref{thm: computing Kan extensions} one can compute the minimal extension ($\underline{M}_K$) and the maximal extension ($\overline{M}_K$) using equations~\eqref{eqn: right Kan extension} and~\eqref{eqn: left Kan extension} respectively:  

\begin{theorem} \begin{enumerate}[(a)]
\item For all $Y \in \Y$, the minimal extension $\underline{M}_K(Y): \Y_
\f \to \leq_{[0,\infty]}$ is given as:
\begin{equation}
\label{eqn: min extension}
 \underline{M}_K(Y) := {\sf lim}( \pi_{Y\downarrow K} M_\f) = \inf \{ M(X) ~|~ Y \to K(X) \in \Y_\f \} 
\end{equation}

(b)For all $Y \in \Y$, the maximal extension $\overline{M}_K(Y): \Y_
\f \to \leq_{[0,\infty]}$ is given as:
\begin{equation}
\label{eqn: max extension}
    \overline{M}_K(Y) := {\sf colim}( \pi_{K \downarrow Y} M_\f) = \sup \{ M(X) ~|~ K(X) \to Y \in \Y_\f \} 
\end{equation} 
\end{enumerate}
\end{theorem}

See Figure~\ref{Fig: extensions schematic}-(a) for a schematic of the minimal and maximal extensions of a monotone.  

Usually for resource theories the codomain of the monotones is $([0,\infty]^\op, \leq_{[0,\infty]}^{\op}) = ([0,\infty]^\op, \geq_{[0,\infty]})$. In the computation of extensions of op-monotones, $\inf$ is flipped to $\sup$ in equation~\eqref{eqn: min extension}, and $\sup$ to be flipped to $\inf$ in equation~\eqref{eqn: max extension}:

\begin{corollary}
Suppose $M: (\X, \X_\f) \to \pCat{[0,\infty]}{\geq}$ be a monotone and $K: (\X, \X_\f) \to (\Y, \Y_\f)$. Then, 

\begin{enumerate}[(a)]
\item For all $Y \in Y$, the minimal extension $\underline{M}_K(Y): \Y_
\f \to \geq_{[0,\infty]}$ is given as:
\begin{equation}
\label{eqn: min rev extension}
 \underline{M}_K(Y) := {\sf colim}( \pi_{Y \downarrow K} M_\f) = \sup \{ M(X) ~|~ Y \to K(X) \in \Y_\f \} 
\end{equation}

(b)For all $Y \in Y$, the maximal extension $\overline{M}_K(Y): \Y_\f \to \geq_{[0,\infty]}$ is given as:
\begin{equation}
\label{eqn: max rev extension}
    \overline{M}_K(Y) := {\sf lim}( \pi_{K \downarrow Y} M_\f) = \inf \{ M(X) ~|~ K(X) \to Y \in \Y_\f \} 
\end{equation} 
\end{enumerate}

\end{corollary}

\begin{proof}
Note that $(\chaos_{[0,\infty}, \geq_{[0,\infty]}) = (\chaos_{[0,\infty]}^\op, \leq_{[0,\infty]}^\op) =: \pCat{[0,\infty]}{\leq}^\op$. Hence, the limits in $\pCat{[0,\infty]}{\leq}$ are the colimits in $\pCat{[0,\infty]}{\leq}^\op$.
\end{proof}

 Figure~\ref{Fig: extensions schematic} - (b) and (c) visualizes the difference in computation of minimal extension of a (regular) monotone and an op-monotone.

Note that equation~\eqref{eqn: min rev extension} is same as \cite[Equation 2]{GoT20} and equation~\eqref{eqn: max rev extension} is same as \cite[Equation 3]{GoT20}. Let us have a closer look at equations~\eqref{eqn: min rev extension} and~\eqref{eqn: max rev extension}. The minimum extension $\underline{M}_K$ assigns to any resource $Y \in \Y_\f$ the value of a resource $X \in \X_\f$ such that the value of $X$ is lowest among the value of all those resources which can be transformed freely to $Y$ under $K$ $(KX \to Y)$. If there exists no such $X \in \X_\f$ such that $KX$ can be transformed to $Y$ using a free transformation,  then $\underline{M}_K(Y) = 0$ (colimit of the empty diagram is the initial object). 

Similarly, the maximal extension $\overline{M}_K$ assigns to any resource $Y \in \Y_\f$ the value of a resource $X' \in \X_\f$ such that the value of $X'$ is the highest among the value of all those resources which $Y$ can be transformed to freely under $K$ $(Y \to KX)$. If there does not exist any such $X \in \X_\f$ which $Y$ can be transformed to using a free transformation, then $\underline{M}_K(Y) = \infty$ (limit of the empty diagram is the terminal object). 

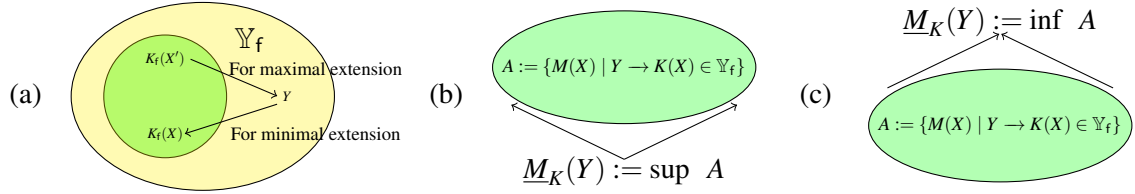
\begin{figure}
    \centering
  (a)~~~  \begin{tikzpicture}[scale=1]
    	\begin{pgfonlayer}{nodelayer}
         \node [style=none, scale=2] (Yf) at (1.25, 3.5) {$\Y_\f$};
         \node[style=none] (Y) at (2.25, 2) {$Y$};
         \node[style=none] (KX) at (-1, 1) {$K_\f(X)$};
         \node[style=none] (KXprime) at (-1, 3) {$K_\f(X')$};
         \node[style=none] (pesudoY) at (2, 2) { };
         \node[style=none] (pesudoYp) at (2, 1.8) { };
		\node [style=none] (pseudoX) at (-0.35, 3) {};
		\node [style=none] (pseudoX') at (-0.5, 1) {};

		\node [style=none,scale=1.3] (3) at (3, 2.75) {For maximal extension};
		\node [style=none, scale=1.3] (4) at (3, 1) {For minimal extension};
	\end{pgfonlayer}
	\begin{pgfonlayer}{edgelayer}
		\draw [->] (pesudoYp) to (pseudoX');
		\draw [->] (pseudoX) to (pesudoY);
	\end{pgfonlayer}
	\begin{pgfonlayer}{backgroundlayer}
		\node [style=circle, scale=9, fill=green, fill opacity=0.4] (0) at (-1, 2) {};
		\node[ellipse,  draw, minimum width = 7cm, 
    minimum height = 5cm, fill=yellow, fill opacity = 0.3] (e) at (0,2) {};
	\end{pgfonlayer}
\end{tikzpicture}~~~
(b)~~~ \begin{tikzpicture}[scale=1]
	\begin{pgfonlayer}{nodelayer}
		\node [style=none] (0) at (-4, 4) {};
		\node [style=none] (1) at (2, 4) {};
		\node [style=none, scale=1.3] (3) at (-1, 5) {$A := \{ M(X) ~|~ Y \to K(X) \in \Y_\f \}$};
		\node [style=none, scale=2] (2) at (-1, 2.25) {$\underline{M}_K(Y) := \sup~~A$};
		\node [style=none, scale=1.2] (4) at (-1, 2.55) {};
	\end{pgfonlayer}
	\begin{pgfonlayer}{edgelayer}
		\draw [<-] (0.center) to (4.center);
		\draw [<-] (1.center) to (4.center);
	\end{pgfonlayer}
	\begin{pgfonlayer}{backgroundlayer}
			\node [style=ellipse, draw, minimum width = 7cm, 
    minimum height = 3cm, fill=green, fill opacity = 0.3] (e) at (-1, 5) {};
	\end{pgfonlayer}
\end{tikzpicture} ~~~~ (c) ~~~ \begin{tikzpicture}[scale=1]
	\begin{pgfonlayer}{nodelayer}
		\node [style=none] (0) at (-4, 3) {};
		\node [style=none] (1) at (2, 3) {};
		\node [style=none, scale=1.3] (3) at (-1, 2) {$A := \{ M(X) ~|~ Y \to K(X) \in \Y_\f \}$};
		\node [style=none, scale=2] (2) at (-1, 4.75) {$\underline{M}_K(Y) := \inf~~A$};
		\node [style=none, scale=1.2] (4) at (-1, 4.45) {};
	\end{pgfonlayer}
	\begin{pgfonlayer}{edgelayer}
		\draw [->] (0.center) to (4);
		\draw [->] (1.center) to (4);
	\end{pgfonlayer}
	\begin{pgfonlayer}{backgroundlayer}
			\node [style=ellipse, draw, minimum width = 7cm, 
    minimum height = 3cm, fill= green, fill opacity = 0.3] (e) at (-1, 2) {};
	\end{pgfonlayer}
\end{tikzpicture}
    \caption{(a) Schematic of minimal and maximal extensions of any monotone along $K$; (b) Minimal extension of monotone $M$; (c) Minimal extension of an op-monotone $M$}
    \label{Fig: extensions schematic}
\end{figure}

We define what it means for the computed extensions to be optimal:
\begin{definition}
\label{defn: optimal}
The minimal extension of a monotone is {\bf optimal} if for any other monotone $G: (\Y, \Y_\f) \to \pCat{[0,\infty]}{\leq}$  such that for all $X \in \X$, $G(K(X)) \leq M(X)$, we have that
\[  G(Y) \leq \underline{M}_K(Y) \]

The maximal extension of a monotone is {\bf optimal} if for any other $G: (\Y, \Y_\f) \to \pCat{[0,\infty]}{\leq}$ such that for all $X \in \X$, $ M(X) \leq G(K(X))$, we have that 
\[ \overline{M}_K(Y) \leq G(Y)  \]
\end{definition}

For the extensions of an op-monotone to be optimal, $``\leq"$ is replaced by $``\geq"$ in definition~\ref{defn: optimal}.

\begin{theorem} 
\label{Lemma: extension properties}
Let $\underline{M}_K$ and $\overline{M}_K$ be minimal and maximal extensions of a monotone $M: (\X, \X_\f) \to \pCat{[0, \infty]}{\leq}$ along a pCat functor $K: (\X, \X_\f) \to (\Y, \Y_\f)$ as per definition~\ref{defn: monotone extensions}. Then, 

\begin{enumerate}[(a)]

\item {\bf Reduction:} For all $X \in \X$, 
\[ \underline{M}_K (K_\f(X))  \leq M(X) \leq \overline{M}_K(K_\f(X)) \]

\item {\bf Monotonicity:} For all $ f: A \to B \in \Y_\f$,  
 \[ \underline{M}_K (A) \leq \underline{M}_K (B)  \text{ and }
 \overline{M}_K (A) \leq \overline{M}_K (B) \]

\item {\bf Optimality:} $\underline{M}_K$ and $\overline{M}_K$ are optimal. 

\end{enumerate}
\end{theorem} 
\begin{proof}~

\begin{enumerate}[(a)]

\item Since $(\underline{M}_K, \leq)$ is the right Kan extension, for all $X \in \X$, $\underline{M}_K(K_\f(X)) \leq M_\f(X) = M(X)$ (see Figure~\ref{Fig: extensions}-(a)).

Since $(\overline{M}_K, \leq)$ is the left Kan extension, for all $X \in \X$, $M(X) = M_\f(X) \leq \overline{M}_K(K_\f(X))$ (see Figure~\ref{Fig: extensions}-(b)).

\item Monotonicity follows from functoriality of $\underline{M}_K$ and $\overline{M}_K$

\item The extensions are optimal by construction (see Figure~\ref{Fig: extensions} for the universal properties). 
\end{enumerate}
\end{proof}

In the above lemma, Statement $(a)$ tells us that the minimal and maximal extensions are respectively a lower and upper bound for $M$ on $\X$. Statement $(b)$ assures that the extensions are monotonic on free transformations. Statement $(c)$ assures that the minimal and maximal extensions are respectively the greatest lower bound and the least upper bound for any other extension of $M$ along $K$, hence are optimal.

\begin{corollary} 
\label{corr: extension properties}
Let $\underline{M}_K$ and $\overline{M}_K$ be minimal and maximal extensions of an op-monotone $M: (\X, \X_\f) \to \pCat{[0, \infty]}{\geq}$ along a pCat functor $K: (\X, \X_\f) \to (\Y, \Y_\f)$ as per definition~\ref{defn: monotone extensions}. Then, the following properties hold for the extensions:

\begin{enumerate}[(a)]

\item {\bf Reduction:} For all $X \in \X$, \[ \underline{M}_K (K(X))  \geq M(X) \geq \overline{M}_K(K(X))\]

\item {\bf Monotonicity:} For all $ f: A \to B \in \Y_\f$, \[  \underline{M}_K (A) \geq \underline{M}_K (B)  \text{ and }
 \overline{M}_K (A) \geq \overline{M}_K (B)  \]

\item {\bf Optimality:} $\underline{M}_K$ and $\overline{M}_K$ are optimal. 
\end{enumerate}
\end{corollary}
\begin{proof}
$\pCat{[0, \infty]}{\geq} = \pCat{[0, \infty]}{\leq}^\op$
\end{proof}

Our Corollary~\ref{corr: extension properties} corresponds to \cite[Theorem 1]{GoT20}. However, in contrast to the proof in \cite{GoT20}, our proof uses only the structural properties of the extensions rather than the formula used to compute them. Moreover, Lemma~\ref{Lemma: extension properties} is more general since in \cite[Theorem 1]{GoT20}, $K: (\X, \X_\f) \to (\Y, \Y_\f)$ is fixed to be a full and faithful inclusion. 

\cite[Theorem 1 - (a)]{GoT20} can be recovered precisely by fixing $K$ to be full and faithful functor in Corollary~\ref{corr: extension properties}:

\begin{corollary}
\label{Lemma: extension properties ff}
If $\underline{M}_F$ and $\overline{M}_F$ are minimal and maximal extensions respectively of an op-monotone $M: (\X, \X_\f) \to \pCat{[0,\infty]}{\geq}$ along a full and faithful (ff) functor $F: (\X, \X_\f) \to (\Y, \Y_\f)$, then: 

\begin{enumerate}[(a)]
\item {\bf ff-Reduction:} The extensions exactly preserves the value of the resources in $\X$ under the action of $F$: 
\[\overline{M}_F(F(X)) = M(X) = \underline{M}_F(F(X)) \]

\item {\bf ff-Optimality:} For any other monotone $G: (\X, \X_\f) \to \pCat{[0,\infty]}{\geq}$, which exactly preserves the value of the resources in $\X$ under the action of $F$, that is, ($ G (F(X)) = M(X)$), then for all $Y \in \Y$:

\[\overline{M}_F(F(Y)) \geq G(Y) \geq \underline{M}_F(F(Y)) \]
\end{enumerate}
\end{corollary}

\begin{proof} 
Statement (a) follows from Lemma~\ref{Lem: full and faithful}. 

For Statement (b), it is given that for all $X \in \X$, $G(F(X)) = F(X)$. From Lemma~\ref{Lemma: extension properties}-(c), it follows that for all $Y \in \Y$,
\[ \underline{M}_F(Y) \geq G(Y) \geq \overline{M}_F(Y) \]

\end{proof}

\subsubsection{ Extending bipartite entanglement monotone from pure to mixed states:}

\begin{example}
We introduced the op-monotone $\Schmidt: (\PureBip, \LOCCp) \to \pCat{[0,\infty]}{\geq}$ in Example~\ref{eg: Schmidt}. Let us extend the monotone from pure bipartite states to mixed states along $i: (\PureBip, \LOCCp) \hookrightarrow (\Bip, \LOCC)$ (defined in Lemma~\ref{lem: Bip inclusion}), something that was already done in \cite{TeP00}, however without the general machinery for computing extensions. 

Figure~\ref{Fig: Bip extensions} presents the diagrams corresponding to minimal and maximal extensions of the monotone $\Schmidt$ along the inclusion. The minimal and the maximal extensions are computed using equations~\eqref{eqn: min rev extension} and~\eqref{eqn: max rev extension} respectively.

It was pointed out in \cite{GoT20} that the definition for the Schmidt entanglement monotone on mixed bipartite states introduced in \cite{TeP00} coincides with equation~\eqref{eqn: max rev extension} referring to the maximal extension of $\Schmidt$. 

\begin{figure}[h]
\[{\small (a)~~~  \xymatrixcolsep{3pc} \xymatrix{ 
&  \ar@<35pt>@{=>}[d]^{\geq} &\\
\LOCCp \ar@{^{(}->}[r]^{i} \ar@/_3pc/[rr]_{{\sf Schmidt}_\LOCCp} & \LOCC \ar@/^3pc/[r]^{G}
\ar@{.>}[r]_---{\underline{{\sf Schmidt}}_i} \ar@{=>}[d]^{\geq}  & \geq_{[0,\infty]} \\ 
 & & &
  }
   ~~~~~ 
 (b)~~~  \xymatrix{ 
 &  \ar@<35pt>@{<=}[d]^{\geq} &\\
 \LOCCp \ar@{^{(}->}[r]^{i} \ar@/_3pc/[rr]_{{\sf Schmidt}_\LOCCp} & \LOCC \ar@/^3pc/[r]^{G} \ar@{.>}[r]_---{\overline{{\sf Schmidt}}_i} \ar@{<=}[d]^{\geq}  & \geq_{[0,\infty]} \\ 
 & & &
  }} \]
 \caption{(a) Minimal extension of ${\sf Schmidt}: (\PureBip, \LOCCp)$ along $i: (\PureBip, \LOCCp) \hookrightarrow (\Bip, \LOCC)$; (b) Maximal extension of ${\sf Schmidt}: (\PureBip, \LOCCp)$ along $i: (\PureBip, \LOCCp) \hookrightarrow (\Bip, \LOCC)$ }
 \label{Fig: Bip extensions}
\end{figure}
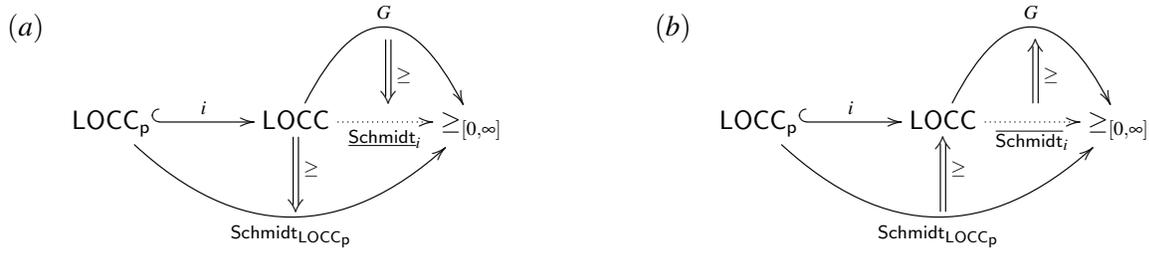
\end{example}

\subsubsection{Extending classical divergences}

Next we examine the properties of extensions of classical divergences to quantum setting:

\begin{lemma}
Let $D: (\cDivergence, \cProcessing) \to \pCat{[0, \infty]}{\geq}$ be a classical divergence as defined in Definition~\ref{defn: divergence monotone}. Let $\underline{D}_i$ and $\overline{D}_i$ be the minimal and maximal extensions respectively of $D$ along $i: (\cDivergence, \cProcessing) \hookrightarrow (\Divergence, \Processing)$. Then the extensions satisfy the following properties:

\begin{enumerate}[(a)]

\item {\bf Reduction}: For all $((p, q),X) \in \Divergence$, $\underline{D}_i(p || q) = D(p || q) = \overline{D}_i(p || q)$

\item {\bf Monotonicity}: For any $M: ((p,q),X) \to ((p', q'), Y)$, $\underline{D}_i(p || q) \geq \underline{D}_i(pM || qM)$ and 
$\overline{D}_i(p || q) \geq \overline{D}_i(pM || qM)$

\item {\bf Optimality}: Suppose $D': (\Divergence, \Processing) \rightarrow \pCat{[0,\infty]}{\geq}$ is any pCat functor such that  for all $((p,q),X) \in \cDivergence$, $D'(i((p,q),X)) = D(((p,q),X))$. Then for all $ ((\rho, \sigma), H) \in \Divergence$,
\begin{equation}
\label{eqn:optimal div}
    \underline{D}_i(\rho || \sigma) \geq G(\rho || \sigma) \geq \overline{D}_i( \rho || \sigma)
\end{equation}  
\end{enumerate}
\end{lemma}
\begin{proof}
By Lemma~\ref{Lemma: Lemma: Div inclusion ff}, the inclusion $i: \cDivergence \hookrightarrow \Divergence$ is full and faithful. Hence, statement (a) and Statement (c) follows directly from ff-Reduction and ff-Optimality properties respectively in Lemma~\ref{Lemma: extension properties}. Statement (b) follows from Monotonicity property in Lemma~\ref{Lemma: extension properties ff}-(b). 
\end{proof}

In the above statement, by the {\bf Reduction} property, $\underline{D}_i$ and $\overline{D}_i$ reduces to classical divergence $D$ on the classical states (pairs of density matrices with off diagonal elements to be zero). The {\bf optimality} property ensures that, for any other quantum divergence that coincides with $D$ on the classical states, must lie between the maximal and minimal extensions in the sense of Eqn.~\eqref{eqn:optimal div}.

\subsubsection{Extending Shannon entropy}

Now we show that Kan extensions are related to some proposals of extending Shannon entropy from classical states to states of a general physical theory \cite{Entropy-Barnum,Entropy-Short,Entropy-Kimura,ScandoloPhD,TowardsThermo,Colleagues}. Specifically, a measurement and a preparation extensions were proposed. Here, for simplicity, we will explain them in the context of quantum theory. In more detail, the measurement
entropy $H_{{\rm meas}}$ of a quantum state $\rho$
is defined as
\begin{equation}
\label{eqn: Hmeas}
H_{{\rm meas}}\left(\rho\right):=\inf_{\boldsymbol{F}}~H\left(q \right),
\end{equation}
where the infimum is taken over all rank-one POVMs  $\boldsymbol{F}:=\left\lbrace F_{j}\right\rbrace $, and $q$ is a probability distribution with $q_{j}:=\mathrm{tr}\:F_j\rho$. Recall that a POVM is a collection of positive semi-definite operators $\left\lbrace F_j\right\rbrace$ that sum to the identity.
On the other hand, the preparation entropy $H_{{\rm prep}}$ is defined as
\begin{equation}
\label{eqn: Hprep}
H_{{\rm prep}}\left(\rho\right):=\inf_{\sum_{j}\lambda_{j}\psi_{j}=\rho}H\left(\lambda \right),
\end{equation}

where the infimum is over all convex decompositions $\sum_{j}\lambda_{j}\psi_{i}$
of the state $\rho$ in terms of pure states $\psi_{j}$ (recall that a quantum state $\psi$ is pure if $\psi^2=\psi$). In words, the measurement entropy $H_{{\rm meas}}$ is the smallest
amount of randomness (as measured by Shannon entropy $H$) present in
the probability distributions generated by rank-one POVMs
on $\rho$. On the other hand, the preparation entropy $H_{{\rm prep}}$
is the smallest amount of randomness necessary to prepare $\rho$ as
an convex combination of pure states.

Let us consider the inclusion (Example~\ref{eg: non-uniformity inclusion})  of resource theory of non-uniformity (given in Example~\ref{eg: RandUni}) into  quantum non-uniformity (given in Example~\ref{eg: qRandUni}). Consider extending the monotone $\Shannon: (\Rand, \Uniform) \to (\chaos_{[0,\infty]}, \leq_{[0,\infty]})$ (given in Example~\ref{eg: Shannon entropy}) along the inclusion as shown in Figure~\ref{Fig: Shannon entopy extensions}. By the Kan extensions formula in equations~\eqref{eqn: min extension} and~\eqref{eqn: max  extension}, the minimal and maximal extension of $\Shannon$ are given as follows:

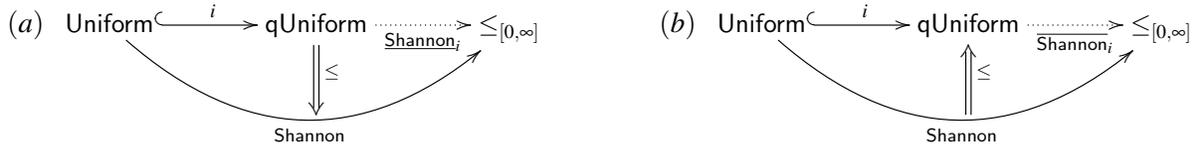
\begin{figure}[ht]
\[{ \small (a)~~  \xymatrixcolsep{3pc} \xymatrix{ 
\Uniform \ar@{^{(}->}[r]^{i} \ar@/_3pc/[rr]_{\Shannon} & \qUniform 
\ar@{.>}[r]_---{\underline{\Shannon}_i} \ar@{=>}[d]^{\leq}  & \leq_{[0,\infty]} \\ 
 & & &
  }
 (b)~~  \xymatrix{ 
 \Uniform \ar@{^{(}->}[r]^{i} \ar@/_3pc/[rr]_{\Shannon} & \qUniform  \ar@{.>}[r]_---{\overline{\Shannon}_i} \ar@{<=}[d]^{\leq}  & \leq_{[0,\infty]} \\ 
 & & &
  } } \]
 \caption{(a) Minimal extension of $\Shannon: (\Rand, \Uniform)$ along $i: (\Rand, \Uniform) \hookrightarrow (\qRand, \qUniform)$; (b) Maximal extension of $\Shannon: (\Rand, \Uniform)$ along $i: (\Rand, \Uniform) \hookrightarrow (\qRand, \qUniform)$ }
 \label{Fig: Shannon entopy extensions}
\end{figure}

\begin{enumerate}[]
\item For all $\rho \in \qUniform$, the minimal extension $\underline{\Shannon}_i: \qUniform \to \leq_{[0,\infty]}$ is given as:
\begin{equation}
\label{eqn: min Shannon extension}
 \underline{\Shannon}_i(\rho) := \inf \{ \Shannon(p) ~|~ \rho \to i(p) \in \qUniform \} 
\end{equation}

For all $\rho \in \qUniform$, the maximal extension $\overline{\Shannon}_i: \qUniform \to \leq_{[0,\infty]}$ is given as:
\begin{equation}
\label{eqn: max Shannon extension}
    \overline{\Shannon}_i(\rho) := \sup \{ \Shannon(p) ~|~ i(p) \to \rho \in \qUniform \} 
\end{equation} 
\end{enumerate}

Let us have a closer look at equation~\eqref{eqn: min Shannon extension} The only unital channels from a quantum to a classical system are given by rank 1 projective measurements $\{P_j \}$, where 
 $P_j$ are rank 1 orthogonal projectors. With this in mind, eqn.~\eqref{eqn: min Shannon extension} can be rewritten as follows:
 
 \[ 
    \underline{\Shannon}_i(\rho) :=\inf_{\boldsymbol{P}}~H\left(q \right),
 \]
 where the infimum is taken over all rank-one projective measurements  $\boldsymbol{P}:=\left\lbrace P_{j}\right\rbrace $, and $q$ is a probability distribution with $q_{j}:=\mathrm{tr}\:P_j\rho$. Now we are going to show that $\underline{\Shannon}_i(\rho)=H_{{\rm meas}}\left(\rho\right)$. To this end, notice that $\underline{\Shannon}_i(\rho)\geq H_{{\rm meas}}\left(\rho\right)$ because the infimum in the definition of $\underline{\Shannon}_i(\rho)$ is over a smaller set. In theorem~5.4.15 of~\cite{ScandoloPhD} it was shown that $ H_{{\rm meas}}\left(\rho\right)$ is achieved by considering the spectral POVM, which is a rank-1 projective measurement. Being $\underline{\Shannon}_i(\rho)$ defined as the infimum over rank-1 projective measurements, then we also have $\underline{\Shannon}_i(\rho)\leq H_{{\rm meas}}\left(\rho\right)$, from which we conclude that $\underline{\Shannon}_i(\rho)=H_{{\rm meas}}\left(\rho\right)$. Since $H_{{\rm meas}}\left(\rho\right)$ is achieved by the spectral measurement, we know that $H_{{\rm meas}}\left(\rho\right)= H\left(p\right)$, where $p$ denotes the classical vector of the spectrum of $\rho$. This shows that $H_{{\rm meas}}$ as defined in equation~\eqref{eqn: Hmeas} is indeed a monotone, as it coincides with the minimal Kan extension. 

Let us now have a closer look at equation~\eqref{eqn: max Shannon extension}. The only unital channels from a classical to a quantum system are given by preparations of a convex combination of pure states $\left\{\psi_j\right\}$ associated with an orthonormal basis of the Hilbert Space corresponding to the quantum system, where the coefficients are the entries of the classical state on which the channel acts. With this in mind, eqn~\eqref{eqn: max Shannon extension} can be rewritten as follows:
 
 \[ 
    \overline{\Shannon}_i(\rho) :=\sup_{\sum_{j}\lambda_{j}\psi_{j}=\rho}H\left(\lambda \right),
 \]
 where the supremum is taken over all decompositions of $\rho$ into \emph{orthogonal} pure states. Now, we we observe that all such decompositions are diagonalizations of $\rho$ (that is, $\rho = \sum_{j} \lambda_j | \psi_j \rangle \langle \psi_j |$ with $\lambda$ being a probability distribution), and therefore they have the same coefficients $\lambda_j$, which are the eigenvalues of $\rho$. In other words,  $\overline{\Shannon}_i(\rho)= H\left(p\right)$, where $p$ denotes the classical vector of the spectrum of $\rho$. Since there is only one vector $\lambda$ (up to permutation) to optimize over, the supremum can be replaced with an infimum. With this in mind, we obtain an expression that is close the preparation entropy.
  \[ 
    \overline{\Shannon}_i(\rho) :=\inf_{\sum_{j}\lambda_{j}\psi_{j}=\rho}H\left(\lambda \right),
 \]
where the infimum is taken over all decompositions of $\rho$ into \emph{orthogonal} pure states. In Theorem~5.4.15 of~\cite{ScandoloPhD} it was shown that $ H_{{\rm prep}}\left(\rho\right)= H\left(p\right)$, from which we have that $\overline{\Shannon}_i(\rho)= H_{{\rm prep}}\left(\rho\right)$. This shows that $H_{{\rm meas}}$ as defined in equation~\eqref{eqn: Hprep} is indeed a monotone, as it coincides with the maximal Kan extension.

Notice that in this example, the minimal and maximal Kan extensions coincide.

\section{Conclusion}

In this article, we studied resource theories as partitioned categories (pCats) and relationship between resource theories as pCat functors thereof. A partitioned category (pCat) is a category with a chosen subcategory  of free transformations. In this framework, a  monotone for a resource theory can be viewed as a pCat functor from the theory into $(\chaos_{[0,\infty]}, \leq_{[0,\infty]})$ where the 
pCat $(\chaos_{[0,\infty]}, \leq_{[0,\infty]})$ represents the partial order $([0,\infty], \leq)$.

We showed that a monotone can be extended from one theory to another using Kan extensions. We applied our framework to extend entanglement monotones for bipartite pure states to bipartite mixed states, to extend classical divergences to the quantum setting, and to extend non-uniformity monotone from classical probabilistic theory to quantum theory.
 
This project was inspired by Gour and Tomamichel's work \cite{GoT20} (see also \cite{Gonda}), which uses a set-based framework to provide formulae for the minimal and maximal extensions of a monotone for a resource theory that embeds (fully and faithfully) in a larger theory. The goal of our work was to present resource theories and monotones in a framework such that the extension formulae for monotones arise naturally. We found that they are precisely given by the well-studied notion of Kan extensions. On top of providing a natural ground to study extensions of monotones, we should also note that our categorical framework is  also more general than the framework in \cite{GoT20}, in that it can be used to compute monotone extensions when the pCat functor between resource theories is not a full and faithful embedding.

\section*{Acknowledgments}
P.\ V.\ S.\ thanks Dr.\ Spencer Breiner, Dr.\ Joe Moeller, and Dr.\ Eswaran Subramanian for  valuable discussions. C.\ M.\ S.\ acknowledges the support of the Natural Sciences and
Engineering Research Council of Canada (NSERC) through the Discovery
Grant ``The power of quantum resources'' RGPIN-2022-03025 and the
Discovery Launch Supplement DGECR-2022-00119.

\bibliographystyle{eptcs}
\bibliography{main}

\end{document}

%% file: tikz.tex
\tikzstyle{strings}=[baseline={([yshift=-.5ex]current bounding box.center)}]

\tikzset{every picture/.append style={scale=.5}, transform shape, strings}

\tikzset{%
symbol/.style={%
draw=none,
every to/.append style={%
edge node={node [sloped, allow upside down, auto=false]{$#1$}}}
}
}

\usetikzlibrary{shapes.geometric}
\usetikzlibrary{patterns}
\usetikzlibrary{fit}
\usetikzlibrary{positioning}
\usetikzlibrary{calc}
\usetikzlibrary{arrows}
\usetikzlibrary{decorations.markings}
\usetikzlibrary{decorations.pathreplacing}
\usetikzlibrary{shapes}

\pgfdeclarelayer{nodelayer}
\pgfdeclarelayer{edgelayer}
\pgfdeclarelayer{backgroundlayer}
\pgfsetlayers{ backgroundlayer, edgelayer,nodelayer, main}

\tikzset{simple/.style={}}
\tikzset{nothing/.style={outer sep=-3.4pt}}
\tikzset{map/.style={draw,fill=white, rectangle}}

\tikzset{dot/.style={thick, fill=black, circle, scale=1, inner sep = .05cm}}

\tikzset{oa/.style={draw, scale=0.9,minimum height=.1cm,circle,append after command={
[shorten >=\pgflinewidth, shorten <=\pgflinewidth,]
(\tikzlastnode.north) edge (\tikzlastnode.south)
(\tikzlastnode.east) edge (\tikzlastnode.west)
}
}
}

\tikzset{ox/.style={draw, rotate=45,scale=0.9,minimum height=.1cm,circle,append after command={
[shorten >=\pgflinewidth, shorten <=\pgflinewidth,]
(\tikzlastnode.north) edge (\tikzlastnode.south)
(\tikzlastnode.east) edge (\tikzlastnode.west)
}
}
}

\tikzset{circ/.style={
shape=circle, inner sep=1pt, draw}}

\tikzstyle{none}=[inner sep=-1pt]
\tikzstyle{circle}=[shape=circle,draw]

\tikzset{wires/.style={}}

\tikzset{box/.style={inner sep=0pt, thick, draw=black, text height=1.5ex, text depth=.25ex, text centered, minimum height=3em, anchor=center}}